\documentclass[preprint]{ptephy_v1}

\preprintnumber{XXXX-XXXX} 
\usepackage{hyperref}

\usepackage{bm}
\usepackage{url}
\usepackage{amsmath}
\usepackage{amsfonts}
\usepackage{mathrsfs}
\usepackage{physics}
\usepackage{ascmac}
\usepackage{fancybox}
\usepackage{ulem}
\usepackage{amsthm}
\usepackage{amssymb}
\usepackage{tikz}
\usepackage{graphicx}
\usepackage{here}
\usepackage{latexsym}
\usepackage{mathtools}

\theoremstyle{plain}

\newtheorem{prop}{Proposition}[section]
\newtheorem{lem}{Lemma}[section]
\theoremstyle{definition}

\newtheorem{cor}{Corollary}[section]

\begin{document}

\title{Advantages of the Kirkwood-Dirac distribution among general quasi-probabilities for finite-state quantum systems}

\author[1]{Shun Umekawa}
\affil[1]{Department of Physics, the University of Tokyo, 7-3-1 Hongo, Bunkyo-ku, Tokyo 113-0033, Japan\email{umeshun2003@g.ecc.u-tokyo.ac.jp}}

\author[2]{Jaeha Lee}
\affil[2]{Institute of Industrial Science, the University of Tokyo, 5-1-5 Kashiwanoha, Kashiwa, Chiba 277-8574, Japan}

\author[2]{Naomichi Hatano}

\begin{abstract} 
We investigate features of the quasi-joint-probability distribution for  finite-state quantum systems, especially the two-state and three-state quantum systems, comparing different types of quasi-joint-probability distributions based on the general framework of quasi-classicalization.
We show from two perspectives that the Kirkwood-Dirac distribution is the quasi-joint-probability distribution that behaves nicely for the finite-state quantum systems.
One is the similarity to the genuine probability and the other is the information that we can obtain from the quasi-probability.
By introducing the concept of the possible values of observables, we show for the finite-state quantum systems that the Kirkwood-Dirac distribution behaves more similarly to the genuine probability distribution in contrast to most of the other quasi-probabilities including the Wigner function.
We also prove that the states of the two-state and three-state quantum systems can be completely distinguished by the Kirkwood-Dirac distribution of only two directions of the spin and point out for the two-state system that the imaginary part of the quasi-probability is essential for the distinguishability of the state.
\end{abstract}

\subjectindex{xxxx, xxx}

\maketitle

\section{Introduction}

Because of the non-commutativity of observables, it is well-known that the genuine joint-probability distribution of multiple observables cannot be defined in quantum theory.
However, in order to obtain an intuitive interpretation of the quantum theory, many attempts have been made to build a quantum analogue of the joint-probability distribution, which is called the \textit{quasi-joint-probability distribution}, or shortly, \textit{quasi-probability} by many pioneers including Wigner and Kirkwood 
\cite{Wigner1932, Kirkwood1933, Dirac1945, Margenau1961, Born1925, Husimi1940, Glauber1963, Sudarshan1963}.
Though the quasi-probabilities had been at first defined in a heuristic manner, attempts to construct the general treatment of the quasi-probabilities were also made.
The unification of some of the quasi-probabilities of the position and momentum was proposed by Cohen \cite{Cohen2005}, and later all the quasi-probabilities mentioned above were indeed unified in the general framework of \textit{quasi-classicalization} by Lee and Tsutsui \cite{Lee2017, Lee2018}.

Perhaps the most famous one of these quasi-probabilities is the Wigner function \cite{Wigner1932}, which is an example of the quasi-joint-probability distribution of the position and the momentum of a particle.
The Wigner function is useful in the sense that its negativity shows the non-classicality \cite{Kenfack2004}, especially the contextuality \cite{Spekkens2008,Booth2022,haferkamp2021} of the state.
It is indeed widely used in the field of quantum optics and quantum computation \cite{Delfosse2015,Raussendorf2017,Maffei2023}.

On the other hand, there is not much research that focused on the quasi-probability for other systems, such as the two-state quantum system.
Finite-state systems has appeared only as a mere example \cite{Margenau1961} and has not explored seriously.
Indeed, it became possible to compare different types of quasi-probabilities systematically only after the unification by Lee and Tsutsui \cite{Lee2017,Lee2018}.
The question here is which quasi-probabilities behave more nicely, which behave more similarly to genuine probability and which give us more information about the state.
This type of argument was conducted by Hofmann \cite{Hofmann2014}, who pointed out that some reasonable condition for the quasi-probability to behave more similarly to the genuine probability requires it to be the Kirkwood-Dirac distribution \cite{Hofmann2014}.
The Kirkwood-Dirac distribution \cite{Kirkwood1933,Dirac1945} is also useful in that it gives a statistical interpretation \cite{Ozawa2011,Morita2013,Lee2017,Lee2018} of the weak value introduced by Aharonov \cite{Aharonov1964,Aharonov1988}, though it is not much better known than the Wigner function.

In the present work, we investigate features of the quasi-joint-probability distribution for finite-state quantum systems, especially the two-state and three-state quantum systems and compare among types of quasi-probabilities.
Through the following argument, we will show from a point of view different from Hofmann's work that the Kirkwood-Dirac distribution is the quasi-joint-probability distribution that behaves nicely in finite-state systems contrary to optical systems.

In Sec.~\ref{Quasi-Joint-Probability Distributions}, we first review the general framework of the quasi-joint-probability distribution and point out a problem that uniquely occurs for finite-state systems. 
We then show the superiority of the Kirkwood-Dirac distribution in the two-state quantum system in Sec.~\ref{Quasi-Probabilities for Two-State Quantum System} and the three-state quantum system in Sec.~\ref{Quasi-Probabilities for Three-State Quantum System}.
In Sec.~\ref{Quasi-Probabilities for General Finite-State Quantum System}, we refer to a general statement for the finite-state quantum systems.

\section{Quasi-Joint-Probability Distributions}
\label{Quasi-Joint-Probability Distributions}
The purpose of the present section is to introduce the concept of the quasi-joint-probability distribution and to review its general framework \cite{Lee2017,Lee2018}.
We first look at the Born rule and functional calculus, which is the trivial case of quasi-classicalization and the quantization in the sense that they can be uniquely defined in the standard quantum theory.
We then define the quasi-joint-probability distribution and the quantization of general physical quantities as an extension of the trivial case by introducing the quasi-joint-spectrum distribution\cite{Lee2018}. 
Through this framework, the duality of quantization and quasi-classicalization together with the arbitrariness of their definition are revealed.
We then look at the real quasi-classicalization, a quasi-classicalization that generates a real-valued quasi-joint-probability distribution with respect to arbitrary states \cite{Lee2018}.
At the end of this section, we will point out an interesting feature of the quasi-joint-probability distribution of finite-state quantum systems.

\subsection{The general framework of quasi-probabilities and quantizations}
\label{The General Framework of Quasi-Probabilities and Quantizations}
In quantum theory, an orthogonal projection \(\hat{E}_A(a)\) onto the eigenspace corresponding to the eigenvalue \(a\) of an observable \(A\) plays an important role.
By using the orthogonal projection, the probability distribution of the value \(a\) of the observable \(A\) with respect to the state described by a density operator \(\rho\) is derived from the Born rule as in
\begin{equation}
\label{Born rule}
P_\rho(a) = \mathrm{Tr}[E_A(a)\,\rho\,].
\end{equation}
The quantization of an observable described as an univariable function of \(A\) is also derived by using the orthogonal projection from the functional calculus:
\begin{equation}
\label{functional calculus}
f(A) = \int_{-\infty}^\infty f(a) E_A(a) \mathrm{d}a.
\end{equation}
Here, \(A\) denotes the operator which is the expressions of an physical quantity in quantum theory and \(a\) denotes the value which is the expression of the physical quantity in classical theory.

Since the mutual orthogonal projection for the combination of non-commutative observables cannot be defined, we cannot naively extend the  method above in order to find a joint-probability distribution of multiple observables or to quantize a physical quantity described as a multivariable function of observables.
In fact, there do not exist any joint-probability distributions of non-commutative observables that satisfy the axioms of the probability:
\begin{align}
\label{axioms of probability}
&(\mathrm{i}) \;\;\;^\forall E\in \mathfrak{B},\;P(E)\in\mathbb{R}_{\geq0},\;\;\;\mbox{where}\;\mathfrak{B}\;\mbox{is a}\;\sigma\mbox{-algebre of the sample space}\;\Omega;\notag\\
&(\mathrm{ii}) \;\;\;P(\Omega) = 1;\\
&(\mathrm{iii}) \;\;\;P\Bigl(\bigcup_{i=1}^\infty E_i\Bigr) = \sum_{i=1}^\infty P(E_i) \;\;\;\mathrm{for\;disjoint\;sets}\;\{E_i\in\mathfrak{B}\}.\notag
\end{align}
However, it is known that we can define a ``quasi-joint-probability distribution", which satisfies only (ii) and (iii) of the axioms of the probability by extending the method above.

As a special case of functional calculus (\ref{functional calculus}), the quantization of \(e^{-isA}\) is given by
\begin{equation}
e^{-isA} = \int_{-\infty}^\infty e^{-isa} E_A(a)da.
\end{equation}
This formula is also understood as the Fourier transform \(\mathcal{F}\) of the projection operator:
\begin{equation}
e^{-isA} = \sqrt{2\pi}\mathcal{F}\bigl(E_A(a)\bigr)(s).
\end{equation}
Therefore, we can conversely consider the orthogonal projection as the inverse Fourier transform of \(e^{-is\hat{A}}\):
\begin{equation}
E_A(a)=\frac{1}{\sqrt{2\pi}}\mathcal{F}^{-1}\bigl(e^{-isA}\bigl)(a).
\end{equation}

As an extension of this formalism, we define the \(hashed\;operator\) \cite{Lee2018}:
\begin{align*}
\label{hashed operator}
\hat{\#}_{\bm{A}}(\bm{s})\coloneqq \{&\mbox{a suitable `mixture' of the `disintegrated' components of}\\
&\mbox{the unitary groups}\;e^{-is\hat{A}_1},...,e^{-is\hat{A}_n}\},\;\;\bm{s}\in\mathbb{R}^n,\;\bm{A}=(A_1,...,A_n)
\stepcounter{equation}\tag{\theequation}
\end{align*}
and the \textit{quasi-joint-spectral distribution} (QJSD) as the inverse Fourier transform of the hashed operator \cite{Lee2018}:
\begin{equation}
\label{QJSD}
\#_{\bm{A}}(\bm{a})\coloneqq\frac{1}{\sqrt{2\pi}^n}\mathcal{F}^{-1}\Bigl(\hat{\#}_{\bm{A}}(\bm{s})\Bigr)(\bm{a}).
\end{equation}
Note that \(\bm{A}\), thus also \(\hat{\bm{A}}\) and \(\bm{a}\), are now multivariable vectors.

The important point here is that the definition of the hashed operator, and thus that of quasi-joint-spectrum distribution, has arbitrariness.
In the case \(\bm{A}=(A,\;B)\), candidates of the hashed operator include
\begin{align}
\label{hashed operatorの例} 
&\hat{\#}^\mathrm{K}_{(A,B)}(s,t) = e^{-isA}e^{-itB},\\
&\hat{\#}^\mathrm{S_\alpha}_{(A,B)}(s,t) = e^{-i\alpha sA}e^{-itB}e^{-i(1-\alpha)sA},\\
&\hat{\#}^\mathrm{M_\alpha}_{(A,B)}(s,t) = \frac{1+\alpha}{2}e^{-isA}e^{-it\hat{B}}+\frac{1-\alpha}{2}e^{-itB}e^{-isA},\\
&\hat{\#}^\mathrm{W}_{(A,B)}(s,t) = e^{-i(sA+tB)} = \lim_{N \to \infty}(e^{-i\frac{s}{N}A}e^{-i\frac{t}{N}B})^N,\\
&\hat{\#}^\mathrm{B}_{(A,B)}(s,t) = \frac{1}{2}\int_{-1}^1 e^{-i\frac{1-k}{2}sA} e^{-itB} e^{-i\frac{1+k}{2}sA} \,\mathrm{d}k,
\end{align}
where K, W, M and B in the superscripts refer to the Kirkwood-Dirac, Wigner, Margenau-Hill and Born-Jordan respectively, while S does not stand for a specific name.
By using quasi-joint-spectrum distributions, the \textit{quasi-joint-probability distribution} is defined analogously to the Born rule \cite{Lee2018}:
\begin{equation}
\label{quasi-joint-probability distributions}
P_{\rho}(\bm{a}) = \mathrm{Tr}[\#_{\bm{A}}(\bm{a})\,\rho\,].
\end{equation}
In particular, the quasi-joint-probability distribution with respect to a pure state is calculated as
\begin{equation}
\label{純粋状態}
P_{\psi}(\bm{a}) = \bra{\psi}\#_{\bm{A}}(\bm{a})\ket{\psi}.
\end{equation}
Quantizations of a physical quantity \(f(\bm{A})\), which is expressed as a multivariable function of \(\bm{A}\), is also derived from QJSD \cite{Lee2018} as in
\begin{equation}
\label{quantization}
f(\bm{A}) = \int_{\mathbb{R}^n} f(\bm{a})\; \#_{\bm{A}}(\bm{a})\;\mathrm{d}^n\bm{a}.
\end{equation}

Since the quasi-joint-probability distribution and the quantization of a physical quantity are defined by using QJSDs, they also have arbitrariness in its definition.
The quasi-joint-probability distribution generated from the hashed operator of the specific form
\begin{equation}
\label{Kirkwood-Dirac distribution の hashed operator}
\hat{\#}_{\bm{A}}^{\mathrm{K}}(\bm{s}) \coloneqq \prod_{k=1}^n e^{-is_kA_k}
\end{equation}
is known as the Kirkwood-Dirac distribution \cite{Lee2018,Kirkwood1933,Dirac1945}:
\begin{equation}
\label{Kirkwood-Dirac distribution}
K_{\rho}(\bm{a}) \coloneqq \mathrm{Tr}[\#^\mathrm{K}_{\bm{A}}(\bm{a})\,\rho\,].
\end{equation}
The quasi-joint-probability distribution generated from the hashed operator of the form
\begin{equation}
\label{Wigner function の hashed operator}
\hat{\#}_{\bm{A}}^{\mathrm{W}}(\bm{s}) \coloneqq e^{-i\sum_{k=1}^n s_kA_k}
\end{equation}
is a generalization of the Wigner function \cite{Lee2018,Wigner1932,Weyl1927}:
\begin{equation}
\label{Wigner function}
W_{\rho}(\bm{a}) \coloneqq \mathrm{Tr}[\#^\mathrm{W}_{\bm{A}}(\bm{a})\,\rho\,].
\end{equation}
Other quasi-joint-probability distributions are also understood as special cases of this formalism.

Although they have arbitrariness in their definition, there are conditions that quasi-joint-probability distributions must satisfy except for (ii) and (iii) of the axioms of the probability.
The marginal probability distribution of \(A_k\) calculated as
\begin{equation}
\label{marginal probability distributions}
P^{A_k}_\rho(a_k)\coloneqq
\int_{\mathbb{R}^{n-1}}P^{\bm{A}}_\rho(\bm{a})\;\mathrm{d}a_1...\mathrm{d}a_{k-1}\mathrm{d}a_{k+1}...\mathrm{d}a_n
\;\;\;(k = 1,2,...,n)
\end{equation}
should express the probability that the observable \(A_k\) takes \(a_k\) in the state \(\rho\), but it is something that can be computed from the standard quantum theory with no arbitrariness:
\begin{equation}
\label{probability of one observable}
P^{A_k}_\rho(a_k) = \mathrm{Tr}[E_{A_k}(a_k)\,\rho\,].
\end{equation}
Therefore the marginal probability distributions (\ref{marginal probability distributions}) must be consistent with Eq.~(\ref{probability of one observable}).
Thanks to the definitions (\ref{hashed operator}), (\ref{QJSD}) and (\ref{quasi-joint-probability distributions}), any quasi-joint-probability distributions naturally fulfill this condition.
Note that distributions which do not satisfy this condition, such as the Husimi function \cite{Husimi1940}, are not included in this definition of quasi-joint-probability distribution, though it can also be understood as an extension of this formalism\cite{Lee2018}.

The duality of quantization and ``quasi-classicalization", which is the map from a density operator to a quasi-joint-probability distribution, is easily understood from the equivalence between two ways of computing the quasi-expectation-value of a physical quantity \(f(\bm{A})\) in the state \(\rho\); see the diagram in Fig.~\ref{双対関係　図}.
One way is to quasi-classicalize the state (that is, to map the state \(\hat{\rho}\) to the quasi-joint-probability distribution \(P_{\rho}(\bm{a})\)) and to compute the expectation value in classical theory as in
\begin{equation}
\label{期待値：古典}
\langle f(\bm{A}) \rangle = \int_{\mathbb{R}^n} f(\bm{a}) P_\rho(\bm{a}) \mathrm{d}^n\bm{a}.
\end{equation}
The other way is to quantize the physical quantity and to compute the expectation value in quantum theory as in
\begin{equation}
\label{期待値：量子}
\langle f(\bm{A}) \rangle = \mathrm{Tr}[f(\bm{A})\,\rho\,].
\end{equation}
The fact that both Eqs.~(\ref{期待値：古典}) and (\ref{期待値：量子}) give the same quasi-expectation value implies the duality of quasi-classicalization and quantization.
\begin{figure}[H]
\centering
\includegraphics[width=0.5\textwidth]{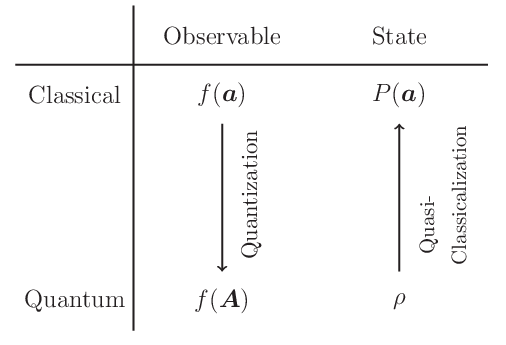}
\caption{The duality of quasi-classicalization and quantization.}
\label{双対関係　図}
\end{figure}

\subsection{Quasi-classicalizations generating real quasi-probabilities}
\label{Quasi-Classicalizations generating Real quasi-probabilities}
Though a quasi-joint-probability distribution generally takes complex values contrary to the regular probability distributions, some quasi-joint-probability distributions stay real.
In this subsection, we will look at which quasi-classicalization generates a quasi-joint-probability distributions that takes only real values.

Since the density operator is generally Hermitian, the statement that a quasi-joint-probability distribution takes real values with respect to arbitrary states, which we refer to as a real quasi-classicalization, is equivalent to the statement that the corresponding QJSD is Hermitian
\begin{equation}
\label{real QJSD}
^\forall\rho,\;P^\mathrm{T}_{\rho}(\bm{a})\in\mathbb{R} \iff \#^\mathrm{T}_{\bm{A}}(\bm{a}) = \#^{\mathrm{T}\dagger}_{\bm{A}}(\bm{a}),
\end{equation}
where \(\mathrm{T}=\mathrm{K,\,W}\) and so on.
Since the Hermitian conjugate of a QJSD is expressed as 
\begin{align*}
\#^{\mathrm{T}\dagger}_{\bm{A}}(\bm{a})
&=\frac{1}{\sqrt{2\pi}^n}\Bigl(\mathcal{F}^{-1}\bigl(\hat{\#}^{\mathrm{T}}_{\bm{A}}(\bm{s})\bigr)(\bm{a})\Bigr)^\dagger\\
&=\frac{1}{(2\pi)^n}\int_{\mathbb{R}^n}\hat{\#}^{\mathrm{T}\dagger}_{\bm{A}}(\bm{s})e^{-i\bm{s}\cdot\bm{a}}\;\mathrm{d}^n\bm{s}\\
&=\frac{1}{(2\pi)^n}\int_{\mathbb{R}^n}\hat{\#}^{\mathrm{T}\dagger}_{\bm{A}}(-\bm{s})e^{i\bm{s}\cdot\bm{a}}\;\mathrm{d}^n\bm{s}\\
&=\frac{1}{\sqrt{2\pi}^n}\mathcal{F}^{-1}\bigl(\hat{\#}^{\mathrm{T}\dagger}_{\bm{A}}(-\bm{s})\bigr)(\bm{a}),
\stepcounter{equation}\tag{\theequation}
\end{align*}
the reality of a quasi-classicalization can be also expressed as the following relation of hashed operators \cite{Lee2018}:
\begin{equation}
\label{real hashed operator}
^\forall\rho,\;P^\mathrm{T}_{\rho}(\bm{a})\in\mathbb{R} 
\iff \#^\mathrm{T}_{\bm{A}}(\bm{a}) = \#^{\mathrm{T}\dagger}_{\bm{A}}(\bm{a}) 
\iff \hat{\#}^{\mathrm{T}}_{\bm{A}}(\bm{s}) = \hat{\#}^{\mathrm{T}\dagger}_{\bm{A}}(-\bm{s}).
\end{equation}
Therefore, we can see that the Wigner distribution (\ref{Wigner function}) always takes real values because Eq.~(\ref{Wigner function の hashed operator}) satisfies the right-most condition of Eq.~(\ref{real hashed operator}), but the Kirkwood-Dirac distribution (\ref{Kirkwood-Dirac distribution}) can take complex values because Eq.~(\ref{Kirkwood-Dirac distribution の hashed operator}) does not satisfy it. 

Real quantization, a quantization corresponding to a real quasi-classicalization, is also meaningful. 
A real quantization quantizes every physical quantity to a Hermitian operators \cite{Lee2018} because
\begin{align}
\Bigl(f^\mathrm{T}(\bm{A})\Bigr)^\dagger 
&=\int_{\mathbb{R}^n} f(\bm{a})\; \#^{\mathrm{T}\dagger}_{\bm{A}}(\bm{a})\;\mathrm{d}^n\bm{a}.\\
&=f^\mathrm{T}(\hat{\bm{A}})
\end{align}
if \(\#^{\mathrm{T}}_{\bm{A}}(\bm{a})\) is a real quantization.

\subsection{Finite-state quantum systems and possible values of observables}
\label{Finit-State Quantum System and Possible Values of Observables}
In quantum theory, the value of an observable is restricted to its eigenvalues. 
Therefore, the value that the observable can take is discretized when we consider a \textit{finite-state quantum system}.
Here, we refer to a quantum system whose state space is a finite-dimensional Hilbert space as a \textit{finite-state quantum system}.
Therefore, the probability distribution obtained from the Born rule (\ref{Born rule}) should take non-zero values only at discrete eigenvalues of the observable in question.
Mathematically speaking, it means that the probability distribution for \(N\)-state quantum systems should be expressed as a linear combination of several delta functions
\begin{equation}
\label{discritization of probability}
P^A_\rho(a)= \sum_{k=1}^n c_k\delta(a-\alpha_k),
\end{equation}
where \(n\leq N\) and \(c_1,...c_n\in\mathbb{R}_{\geq0}\) with \(\sum_{k=1}^n c_k = 1\).
Here, \(\alpha_k\,(k = 1,2,...,n)\) denotes each eigenvalue of the observable \(A\).

However, it is not the case for quasi-joint-probability distributions. Marginal probability distributions obtained from Eq.~(\ref{marginal probability distributions}) should satisfy the conditions (\ref{discritization of probability}), but they are the only restrictions.
Here, we define \((\alpha_{1k_{1}},\alpha_{2k_{2}},...,\alpha_{mk_{m}})\) as \textit{possible values} of observables \(\bm{A}=(A_1,A_2,...,A_m)\) if each \(\alpha_{ik_{i}}\) is an eigenvalue of an observable \(A_i\).
Though we expect the quasi-joint-probability distribution of \(\bm{A}\) to take non-zero values only at the possible value of \(\bm{A}\) in the same way as the genuine probability distribution, not every types of quasi-joint-probability distribution do so as we will see later.
In general, quasi-joint-probability distribution are permitted to take non-zero values everywhere.

\section{Quasi-Probabilities for Two-State Quantum System}
\label{Quasi-Probabilities for Two-State Quantum System}
As seen above, the quasi-joint-probability distribution is determined if we determine what observables to consider, how to quasi-classicalize and the state of the system.
In this section, we will consider quasi-joint-probability distributions for a two-state quantum system, the state space of which is described by a two-dimensional Hilbert space.
At first, we consider the \(x\) and \(y\) components of spin \(1/2\) and look at seven examples.
Through these examples, we will see that some types of quasi-joint-probability distributions behave nicely while others do not depending on the way of quasi-classicalization.
We then generally prove nice features of the Kirkwood-Dirac distributions of spin \(1/2\) and see the reason why the other types of quasi-joint-probability distributions do not show such features.
At the end, we will see that these features of the Kirkwood-Dirac distribution are not limited to the distribution of the \(x\) and \(y\) components of spin.

\subsection{Spin \(1/2\)}
\label{Spin 1/2}
The spin operator \(\{J_i\}\,(i = 1,2,3)\) is generally defined by the commutation relations
\begin{equation}
\label{SU(2)}
[J_i, J_j] =  i\epsilon_{ijk}J_k
\end{equation}
and are understood as a basis set of the Lie algebra \(\mathfrak{su}(2)\).
Here, we use the natural unit \(\hbar=1\).
Since observables are described as \(N\)-dimensional Hermitian matrices for an \(N\)-state quantum system, we look at the representation of \(\mathfrak{s}\mathfrak{u}(2)\) in order to consider the spin in a quantum system.
The \(N\)-dimensional irreducible representation of \(\mathfrak{su}(2)\) is known as the spin \((N-1)/2\) representation.
In spin \(j\) representation, each of the spin operators \(\{J_i\}\) is expressed as a matrix whose eigenvalues are \(j,j-1,...,-j\) and the magnitude of spin \(\bm{J}^2\) becomes \(\mathfrak{su}(2)\)'s Casimir operator
\begin{equation}
\label{J^2}
\bm{J}^2 \coloneqq {J_1}^2 + {J_2}^2 + {J_3}^2  = j(j+1)I.
\end{equation}
For \(N=2\), the irreducible representation of \(\mathfrak{su}(2)\) is the spin \(1/2\) representation. 
In this case, the spin operator \(\{J_i\}\) is particularly described as
\begin{equation}
J_i = \frac{\sigma_i}{2}
\end{equation}
by using the Pauli matrices
\begin{equation}
\label{パウリ行列}
\sigma_1
=\begin{pmatrix}
   0 & 1 \\
   1 & 0
\end{pmatrix},\;
\sigma_2
=\begin{pmatrix}
   0 & -i \\
   i & 0
\end{pmatrix},\;
\sigma_3
=\begin{pmatrix}
   1 & 0 \\
   0 & -1
\end{pmatrix}.
\end{equation}
As discussed in Subsec.~\ref{Finit-State Quantum System and Possible Values of Observables}, the value that an observable of a finite-dimensional quantum system can take is discritized to its eigenvalues.
In the case of the spin of the two-state quantum system, these values are the eigenvalues of the spin operators of the spin \(1/2\) representation, namely \(1/2\) and \(-1/2\).

The states of the two-state quantum systems are described not only by the two-dimensional density matrix, but by the combination of the expectation values in the three directions of the spin, \((\langle J_1\rangle,\,\langle J_2\rangle,\,\langle J_3\rangle)\).
More specifically, any state can be expressed uniquely as a point in or on the Bloch sphere as shown in Fig.~\ref{ブロッホ球　図}.
\begin{figure}[H]
\centering
\includegraphics[width=0.3\textwidth]{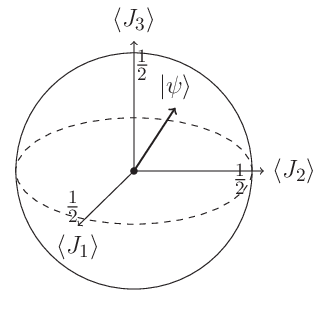}
\caption{Bloch sphere. Any states of spin \(1/2\) is specified as a point \((\langle J_1\rangle,\,\langle J_2\rangle,\,\langle J_3\rangle)\).}
\label{ブロッホ球　図}
\end{figure}

\subsection{Examples}
\label{Examples}
In this subsection, we look at examples of quasi-joint-probability distributions of the \(x\) and \(y\) components of spin \(1/2\) for two-state quantum systems.
First, the exponential functions of \(J_1\) and \(J_2\) of spin \(1/2\) are given by
\begin{equation}
\label{J1の指数}
e^{-isJ_1}
=\begin{pmatrix}
    \cos \frac{s}{2}   & -i\sin \frac{s}{2}\\
    -i\sin \frac{s}{2} & \cos \frac{s}{2}
 \end{pmatrix},
\end{equation}
\begin{equation}
\label{J2の指数}
e^{-itJ_2}
=\begin{pmatrix}
    \cos \frac{t}{2} & -\sin \frac{t}{2}\\
    \sin \frac{t}{2} & \cos \frac{t}{2}
 \end{pmatrix}.
\end{equation}
Then, some of the hashed operators are found as follows:
\begin{equation}
\label{＃K}
\hat{\#}^{\mathrm{K}}_{J_1,J_2}(s,t) 
\coloneqq e^{-isJ_1}e^{-itJ_2}=
\begin{pmatrix}
    \cos \frac{s}{2} \cos \frac{t}{2} -i\sin \frac{s}{2} \sin \frac{t}{2} 
    & -\cos \frac{s}{2} \sin \frac{t}{2} -i\sin \frac{s}{2} \cos \frac{t}{2}\\
    \cos \frac{s}{2} \sin \frac{t}{2} -i\sin \frac{s}{2} \cos \frac{t}{2} 
    & \cos \frac{s}{2} \cos \frac{t}{2} +i\sin \frac{s}{2} \sin \frac{t}{2}
\end{pmatrix},
\end{equation}
\begin{equation}
\label{＃S}
\hat{\#}^{\mathrm{S}_{1/2}}_{J_1,J_2}(s,t) 
\coloneqq e^{-i\frac{s}{2}J_1} e^{-itJ_2} e^{-i\frac{s}{2}J_1}
=\begin{pmatrix}
   \cos \frac{s}{2} \cos \frac{t}{2} 
   & -\sin \frac{t}{2} -i\sin \frac{s}{2} \cos \frac{t}{2}\\
   \sin \frac{t}{2} -i\sin \frac{s}{2} \cos \frac{t}{2} 
   & \cos \frac{s}{2} \cos \frac{t}{2}
 \end{pmatrix},
\end{equation}
\begin{align*}
\label{＃M}
\hat{\#}^{\mathrm{M}_0}_{J_1,J_2}(s,t) 
&\coloneqq \frac{1}{2} (e^{-isJ_1}e^{-itJ_2} + e^{-itJ_2}e^{-isJ_1})\\
&=
\begin{pmatrix}
\cos \frac{s}{2} \cos \frac{t}{2} & -\cos \frac{s}{2} \sin \frac{t}{2} -i\sin \frac{s}{2} \cos \frac{t}{2}\\
\cos \frac{s}{2} \sin \frac{t}{2} -i\sin \frac{s}{2} \cos \frac{t}{2} & \cos \frac{s}{2} \cos \frac{t}{2}
\end{pmatrix},
\stepcounter{equation}\tag{\theequation}
\end{align*}
\begin{align*}
\label{＃W}
\hat{\#}^\mathrm{W}_{J_1,J_2}(s,t) 
&\coloneqq e^{-i(sJ_1+tJ_2)}\\
&=\begin{pmatrix}
    \cos \sqrt{(\frac{s}{2})^2+(\frac{t}{2})^2} 
    & -i\sqrt{\frac{(\frac{s}{2})-i(\frac{t}{2})}{(\frac{s}{2})+i(\frac{t}{2})}} \sin \sqrt{(\frac{s}{2})^2+(\frac{t}{2})^2}\\
    -i\sqrt{\frac{(\frac{s}{2})+i(\frac{t}{2})}{(\frac{s}{2})-i(\frac{t}{2})}} \sin \sqrt{(\frac{s}{2})^2+(\frac{t}{2})^2} 
    & \cos \sqrt{(\frac{s}{2})^2+(\frac{t}{2})^2}
 \end{pmatrix},
\stepcounter{equation}\tag{\theequation}
\end{align*}
\begin{align*}
\label{＃B}
\hat{\#}^\mathrm{B}_{J_1,J_2}(s,t)
&\coloneqq \frac{1}{2}\int_{-1}^1 e^{-i\frac{1-k}{2}sJ_1} e^{-itJ_2} e^{-i\frac{1+k}{2}sJ_1} \,\mathrm{d}k\\
&=
\begin{pmatrix}
\cos \frac{s}{2} \cos \frac{t}{2} & \frac{\sin \frac{s}{2} }{\frac{s}{2}} \sin \frac{t}{2} -i\sin \frac{s}{2} \cos \frac{t}{2}\\
-\frac{\sin \frac{s}{2} }{\frac{s}{2}} \sin \frac{t}{2} -i\sin \frac{s}{2} \cos \frac{t}{2} & \cos \frac{s}{2} \cos \frac{t}{2}
\end{pmatrix}.
\stepcounter{equation}\tag{\theequation}
\end{align*}
By using these expressions, we can compute each type of quasi-joint-probability distributions with respect to an arbitrary state.

For example, using Eq.~(\ref{＃K}) the Kirkwood-Dirac distribution with respect to the \(\ket{z+}\) state, which is the eigenstate of the eigenvalue \(1/2\) of \(J_3\), is given by
\begin{alignat}{2}
\label{Kz+}
K_{z+}(x,y) 
&= \frac{1}{2\pi} \mathscr{F}^{-1}\Bigl(\bra{z+} \hat{\#}^K_{J_1,J_2} \ket{z+}\Bigr)(x,y)\notag\\
&= \frac{1}{2\pi} \mathscr{F}^{-1}\qty(
\begin{pmatrix}
1 & 0
\end{pmatrix}
\begin{pmatrix}
    \cos \frac{s}{2} \cos \frac{t}{2} -i\sin \frac{s}{2} \sin \frac{t}{2} 
    & -\cos \frac{s}{2} \sin \frac{t}{2} -i\sin \frac{s}{2} \cos \frac{t}{2}\\
    \cos \frac{s}{2} \sin \frac{t}{2} -i\sin \frac{s}{2} \cos \frac{t}{2} 
    & \cos \frac{s}{2} \cos \frac{t}{2} +i\sin \frac{s}{2} \sin \frac{t}{2}
\end{pmatrix}
\begin{pmatrix}
1\\
0
\end{pmatrix})\notag\\
&= \frac{1}{4}\qty[\delta\qty(x-\frac{1}{2}) + \delta\qty(x+\frac{1}{2})]\qty[\delta\qty(y-\frac{1}{2}) + \delta\qty(y+\frac{1}{2})]\notag\\
&\;\;\;\;\;\;\;\;\;\;\;\;\;\;\;+ \frac{i}{4}\qty[\delta\qty(x-\frac{1}{2}) - \delta\qty(x+\frac{1}{2})]\qty[\delta\qty(y-\frac{1}{2}) - \delta\qty(y+\frac{1}{2})];
\end{alignat}
see Fig.~\ref{z+-状態のKirkwood-Dirac分布　図} (a) and (b).
In the same way, the one for the \(\ket{z-}\) state, the eigenstate with eigenvalue \(-1/2\) of \(\hat{J}_3\), is obtained in the form
\begin{align*}
\label{Kz-}
K_{z-}(x,y) 
&= \frac{1}{2\pi} \mathscr{F}^{-1}\qty(
\begin{pmatrix}
0 & 1
\end{pmatrix}
\begin{pmatrix}
    \cos \frac{s}{2} \cos \frac{t}{2} -i\sin \frac{s}{2} \sin \frac{t}{2} 
    & -\cos \frac{s}{2} \sin \frac{t}{2} -i\sin \frac{s}{2} \cos \frac{t}{2}\\
    \cos \frac{s}{2} \sin \frac{t}{2} -i\sin \frac{s}{2} \cos \frac{t}{2} 
    & \cos \frac{s}{2} \cos \frac{t}{2} +i\sin \frac{s}{2} \sin \frac{t}{2}
\end{pmatrix}
\begin{pmatrix}
0\\
1
\end{pmatrix})\\
&= \frac{1}{4}\qty[\delta\qty(x-\frac{1}{2}) + \delta\qty(x+\frac{1}{2})]\qty[\delta\qty(y-\frac{1}{2}) + \delta\qty(y+\frac{1}{2})]\notag\\
&\;\;\;\;\;\;\;\;\;\;\;\;\;\;\;- \frac{i}{4}\qty[\delta\qty(x-\frac{1}{2}) - \delta\qty(x+\frac{1}{2})]\qty[\delta\qty(y-\frac{1}{2}) - \delta\qty(y+\frac{1}{2})];
\stepcounter{equation}\tag{\theequation}
\end{align*}
see Fig.~\ref{z+-状態のKirkwood-Dirac分布　図} (c) and (d).
We can see that the Kirkwood-Dirac distribution  with respect to \(\ket{z\pm}\) takes complex numbers and its real part is common to the \(\ket{z\pm}\) states while its imaginary part is different for them.
\begin{figure}[H]
\centering
\includegraphics[width=0.8\textwidth]{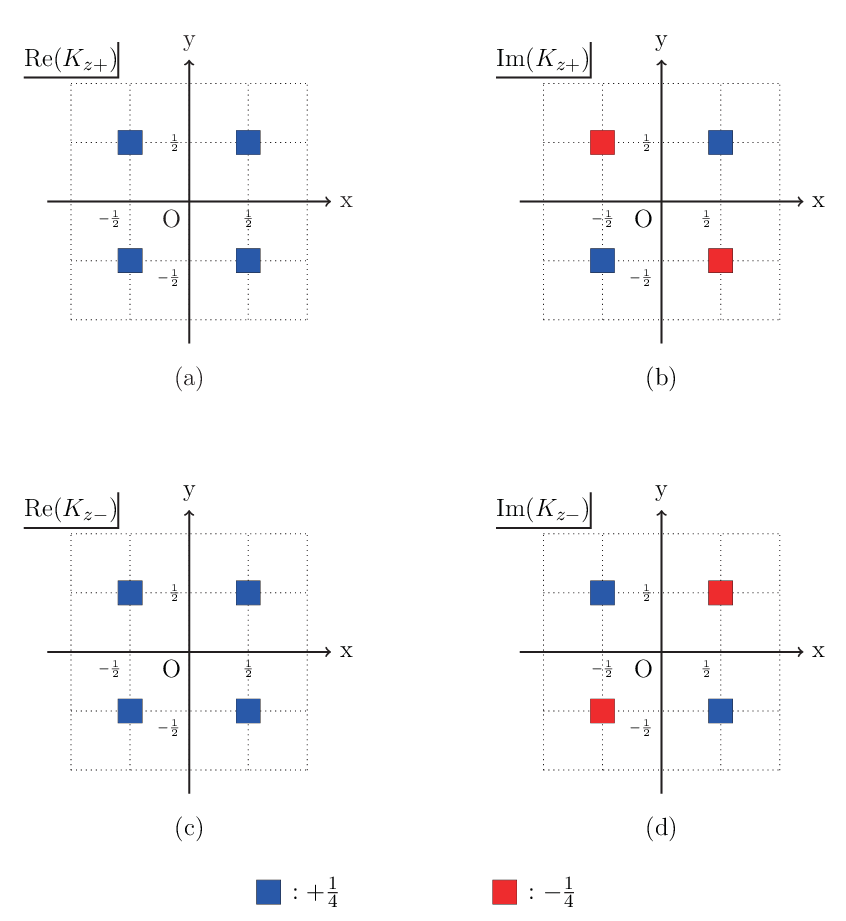}
\caption{The real and imaginary parts of the Kirkwood-Dirac distribution for the \(\ket{z+}\) and \(\ket{z-}\) states. The blue square indicates the value \(+1/4\), while the red square indicates the value \(-1/4\).}
\label{z+-状態のKirkwood-Dirac分布　図}
\end{figure}
On the other hand, the quasi-joint-probability distributions with respect to \(\ket{z+}\) and \(\ket{z-}\) produced by \(\hat{\#}^{S_{1/2}}_{J_1,J_2}(s,t) \) in Eq.~(\ref{＃S}) are respectively given by
\begin{equation}
\label{Sz+}
S_{z+}(x,y)=\frac{1}{4}\qty[\delta\qty(x-\frac{1}{2}) + \delta\qty(x+\frac{1}{2})]\qty[\delta\qty(y-\frac{1}{2}) + \delta\qty(y+\frac{1}{2})],
\end{equation}
\begin{equation}
\label{Sz-}
S_{z-}(x,y)=\frac{1}{4}\qty[\delta\qty(x-\frac{1}{2}) + \delta\qty(x+\frac{1}{2})]\qty[\delta\qty(y-\frac{1}{2}) + \delta\qty(y+\frac{1}{2})];
\end{equation}
see Fig.~\ref{z+-状態のS関数　図}.
Both distributions with respect to \(\ket{z+}\) and \(\ket{z-}\) states are real and they are equal to each other.
While we can distinguish the \(\ket{z\pm}\) states by the Kirkwood-Dirac distribution as in Fig.~\ref{z+-状態のKirkwood-Dirac分布　図}, we cannot by the quasi-joint-probability distribution that generated from \(\hat{\#}^{S_{1/2}}_{J_1,J_2}(s,t)\) as in Fig.~\ref{z+-状態のS関数　図}.
Whether we can distinguish the state by a quasi-joint-probability distribution depends on the choice of quasi-joint-spectrum distribution, which determines how to quasi-classicalize.
\begin{figure}[H]
\centering
\includegraphics[width=0.8\textwidth]{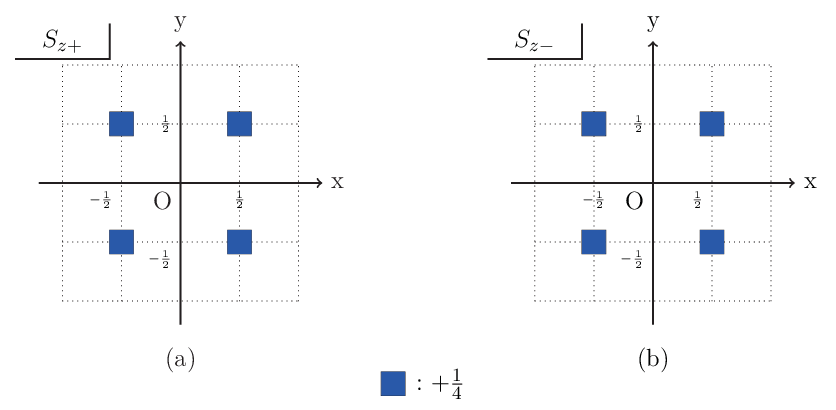}
\caption{Quasi-joint-probability distributions generated from \(\hat{\#}^{S_{1/2}}_{J_1,J_2}(s,t)\) with respect to (a) the \(\ket{z+}\) state and (b) the \(\ket{z-}\) state. Each blue square indicates the value \(+1/4\)}
\label{z+-状態のS関数　図}
\end{figure}
Let us next look at the \(\ket{y+}\) state, the eigenstate of the eigenvalue \(1/2\) of \(J_2\).
The Kirkwood-Dirac distribution and the quasi-joint-probability distribution generated from \(\hat{\#}^{S_{1/2}}_{J_1,J_2}(s,t)\) are respectively given by
\begin{align*}
\label{Ky+}
& K_{y+}(x,y)\\
&= \frac{1}{2\pi} \mathscr{F}^{-1}\qty(
\begin{pmatrix}
\frac{1}{\sqrt{2}} & -\frac{i}{\sqrt{2}}
\end{pmatrix}
\begin{pmatrix}
    \cos \frac{s}{2} \cos \frac{t}{2} -i\sin \frac{s}{2} \sin \frac{t}{2} 
    & -\cos \frac{s}{2} \sin \frac{t}{2} -i\sin \frac{s}{2} \cos \frac{t}{2}\\
    \cos \frac{s}{2} \sin \frac{t}{2} -i\sin \frac{s}{2} \cos \frac{t}{2} 
    & \cos \frac{s}{2} \cos \frac{t}{2} +i\sin \frac{s}{2} \sin \frac{t}{2}
\end{pmatrix}
\begin{pmatrix}
\frac{1}{\sqrt{2}}\\
\frac{i}{\sqrt{2}}
\end{pmatrix})\\
&= \frac{1}{2}\qty[\delta\qty(x-\frac{1}{2}) + \delta\qty(x+\frac{1}{2})]\delta\qty(y-\frac{1}{2}),
\stepcounter{equation}\tag{\theequation}
\end{align*}
\begin{align*}
\label{Sy+}
S_{y+}(x,y)
= &\frac{1}{4}\qty[\delta\qty(x-\frac{1}{2}) + \delta\qty(x+\frac{1}{2})]\qty[\delta\qty(y-\frac{1}{2}) + \delta\qty(y+\frac{1}{2})] \\
&+\frac{1}{2}\delta(x)\qty[\delta\qty(y-\frac{1}{2}) - \delta\qty(y+\frac{1}{2})];
\stepcounter{equation}\tag{\theequation}
\end{align*}
see Fig.~\ref{y+状態のKirkwood-Dirac関数とS関数　図}.
For the \(\ket{y+}\) state, both the Kirkwood-Dirac distribution and the quasi-joint-probability distribution generated from \(\hat{\#}^{S_{1/2}}_{J_1,J_2}(s,t)\) are real, but while the former takes non-zero values only at the pair of possible values of the \(x\) and \(y\) components of the spin \(1/2\), namely \((\pm1/2,\pm1/2)\), the latter takes non-zero values at\((J_1,J_2)=(0,\pm1/2)\) even though the value of \(J_1\) cannot be zero.
\begin{figure}[H]
\centering
\includegraphics[width=0.8\textwidth]{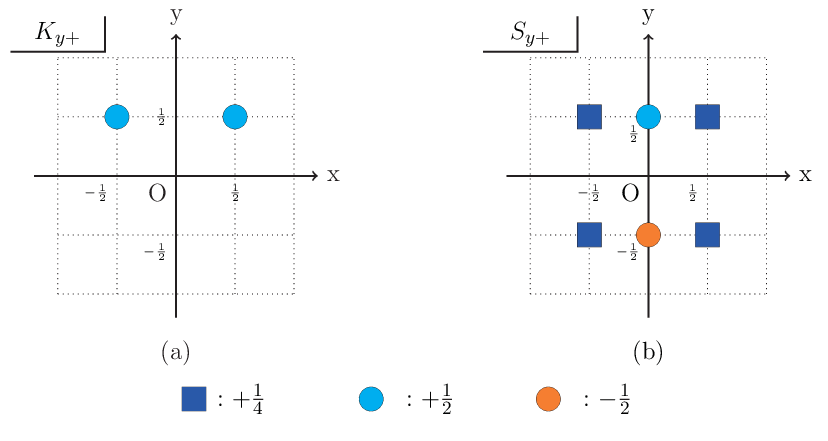}
\caption{(a) the Kirkwood-Dirac distribution and (b) the quasi-joint-probability distribution generated from \(\hat{\#}^{S_{1/2}}_{J_1,J_2}(s,t)\) with respect to the \(\ket{y+}\)  state. The blue square indicates the value \(+1/4\), while the light blue circle indicates the value \(+1/2\) and the orange circle indicates the value \(-1/2\).}
\label{y+状態のKirkwood-Dirac関数とS関数　図}
\end{figure}
Finally, the Wigner distribution, the quasi-joint-probability distribution generated from \(\hat{\#}^W_{J_1,J_2}(s,t)\) in Eq.~(\ref{＃W}), cannot distinguish the \(\ket{z\pm}\) states and even diverge everywhere rather than to take non-zero values only at \((\pm1/2,\,\pm1/2)\):
\begin{align*}
\label{Wz+-}
W_{z\pm}(x,y) 
&= \frac{1}{2\pi} \mathscr{F}^{-1}\Bigl(\bra{z\pm} \hat{\#}^W_{J_1,J_2} \ket{z\pm}\Bigr)(x,y)\\
&= \frac{1}{2\pi} \int_{-\infty}^\infty \int_{-\infty}^\infty  \cos \sqrt{(\frac{s}{2})^2+(\frac{t}{2})^2}e^{isx} e^{ity} \,\frac{\mathrm{d}s \,\mathrm{d}t}{2\pi}\\
&= \frac{1}{(2\pi)^2} \int_{0}^\infty \int_{0}^{2\pi} \cos \frac{r}{2}\; e^{ixr\cos\theta}\; e^{iyr\sin\theta}\; r\, \mathrm{d}\theta\, \mathrm{d}r\\
&= \frac{1}{2\pi} \int_{0}^\infty r \cos \frac{r}{2}\; J_0\bigl(r\sqrt{x^2+y^2}\bigr)\;  \mathrm{d}r\\
&\simeq \frac{1}{\sqrt{2\pi^3}} \int_0^\infty \sqrt{r'} \cos \frac{r'}{2\sqrt{x^2+y^2}}\; \cos(r'-\pi/4)\; \mathrm{d}r',
\stepcounter{equation}\tag{\theequation}
\end{align*}
which can oscillate rapidly between \(+\infty\) and \(-\infty\).
Through the observation of these examples of quasi-joint-probability distributions, we can see that the behavior is substantially different for each quasi-joint probability distribution.
We should look for a quasi-joint-spectrum distribution in order to obtain a quasi-joint-probability distribution that behaves well.

\subsection{General features of Kirkwood-Dirac distributions of spin \(1/2\)}
\label{General Features of Kirkwood-Dirac Distributions of Spin 1/2}
As was exemplified in Subsec.~\ref{Examples}, the Kirkwood-Dirac distribution behaves nicely for the two-state quantum system. 
In the present subsection, we prove three nice features of the Kirkwood-Dirac distributions of the \(x\) and \(y\) components of the spin \(1/2\) for the two-state quantum systems.

First, we look at the possible values of observables.
\begin{prop}
\label{取り得る値　２次元}
The Kirkwood-Dirac distribution of the \(x\) and \(y\) components of the spin \(1/2\) with respect to an arbitrary state of the two-state quantum systems takes non-zero values only at \((J_1,J_2)=(\pm1/2,\pm1/2)\):
\begin{equation}
^\forall\rho,\; K_{\rho}(x,y) = 0\;\;\bigl(\mbox{for}\;(x,y)\neq(\pm1/2,\pm1/2)\bigl).
\end{equation}
\end{prop}
\begin{proof}
The Kirkwood-Dirac distribution of the \(x\) and \(y\) components of spin \(1/2\) for the two-state quantum systems is a quasi-joint-probability distribution generated from the hashed operator \(\hat{\#}^K_{J_1,J_2}(s,t)\) in Eq.~(\ref{＃K}). 
Quasi-joint-spectrum distribution corresponding to this hashed operator is given by
\begin{align*}
\#^K_{J_1,J_2}(x,y)
&=\frac{1}{2\pi}\mathcal{F}^{-1}\qty(\hat{\#}^K_{J_1,J_2}(s,t))(x,y).\\
&=\frac{1}{2\pi}\mathcal{F}^{-1}\qty(
\begin{pmatrix}
\cos \frac{s}{2} \cos \frac{t}{2} -i\sin \frac{s}{2} \sin \frac{t}{2} 
& -\cos \frac{s}{2} \sin \frac{t}{2} -i\sin \frac{s}{2} \cos \frac{t}{2}\\
\cos \frac{s}{2} \sin \frac{t}{2} -i\sin \frac{s}{2} \cos \frac{t}{2} 
& \cos \frac{s}{2} \cos \frac{t}{2} +i\sin \frac{s}{2} \sin \frac{t}{2}
\end{pmatrix}
)(x,y),
\stepcounter{equation}\tag{\theequation}
\end{align*}
whose components read
\begin{align*}
\label{＃K11}
\bigl(\#^\mathrm{K}_{J_1,J_2}(x,y)\bigr)_{11}
&= \frac{1}{4}\qty[\delta\qty(x-\frac{1}{2}) + \delta\qty(x+\frac{1}{2})]\qty[\delta\qty(y-\frac{1}{2}) + \delta\qty(y+\frac{1}{2})]\\
&\;\;\;\;\;\;\;\;\;\;\;\;\;\;\;+ \frac{i}{4}\qty[\delta\qty(x-\frac{1}{2}) - \delta\qty(x+\frac{1}{2})]\qty[\delta\qty(y-\frac{1}{2}) - \delta\qty(y+\frac{1}{2})],
\stepcounter{equation}\tag{\theequation}
\end{align*}
\begin{align*}
\label{＃K12}
\bigl(\#^\mathrm{K}_{J_1,J_2}(x,y)\bigr)_{12}
&= \frac{1}{4}\qty[\delta\qty(x-\frac{1}{2}) - \delta\qty(x+\frac{1}{2})]\qty[\delta\qty(y-\frac{1}{2}) + \delta\qty(y+\frac{1}{2})]\\
&\;\;\;\;\;\;\;\;\;\;\;\;\;\;\;- \frac{i}{4}\qty[\delta\qty(x-\frac{1}{2}) + \delta\qty(x+\frac{1}{2})]\qty[\delta\qty(y-\frac{1}{2}) - \delta\qty(y+\frac{1}{2})],
\stepcounter{equation}\tag{\theequation}
\end{align*}
\begin{align*}
\label{＃K21}
\bigl(\#^\mathrm{K}_{J_1,J_2}(x,y)\bigr)_{21}
&= \frac{1}{4}\qty[\delta\qty(x-\frac{1}{2}) - \delta\qty(x+\frac{1}{2})]\qty[\delta\qty(y-\frac{1}{2}) + \delta\qty(y+\frac{1}{2})]\\
&\;\;\;\;\;\;\;\;\;\;\;\;\;\;\;+ \frac{i}{4}\qty[\delta\qty(x-\frac{1}{2}) + \delta\qty(x+\frac{1}{2})]\qty[\delta\qty(y-\frac{1}{2}) - \delta\qty(y+\frac{1}{2})],
\stepcounter{equation}\tag{\theequation}
\end{align*}
\begin{align*}
\label{＃K22}
\bigl(\#^\mathrm{K}_{J_1,J_2}(x,y)\bigr)_{22}
&= \frac{1}{4}\qty[\delta\qty(x-\frac{1}{2}) + \delta\qty(x+\frac{1}{2})]\qty[\delta\qty(y-\frac{1}{2}) + \delta\qty(y+\frac{1}{2})]\\
&\;\;\;\;\;\;\;\;\;\;\;\;\;\;\;- \frac{i}{4}\qty[\delta\qty(x-\frac{1}{2}) - \delta\qty(x+\frac{1}{2})]\qty[\delta\qty(y-\frac{1}{2}) - \delta\qty(y+\frac{1}{2})].
\stepcounter{equation}\tag{\theequation}
\end{align*}
Since all the components of this quasi-joint-spectrum distribution matrices are expressed as a linear combination of \(\delta(x\pm1/2)\delta(y\pm1/2)\), the Kirkwood-Dirac distribution, the quasi-joint-probability distribution quasi-classicalised by this quasi-joint-spectrum distribution, is also a linear combination of \(\delta(x\pm1/2)\delta(y\pm1/2)\) whatever the components of density operator describing the state is.
\end{proof}
This proposition indicates that the Kirkwood-Dirac distribution is a good quasi-joint-probability distribution in that it takes non-zero-values only at the possible values of the observables in question similarly to the genuine probability distributions, contrary to many other quasi-joint-probability distributions seen in Subsec.~\ref{Examples}.

The second topic is whether we can completely distinguish the state by a quasi-joint-probability distribution.
\begin{prop}
\label{状態を区別できる　２次元スピン}
The state of a two-state quantum system is completely distinguishable by the Kirkwood-Dirac distribution:
\begin{equation}
\rho=\rho '
\iff ^\forall (x,y) \in \mathbb{R}^2,K_{\rho}(x,y)=K_{\rho'}(x,y).
\end{equation}
\end{prop}
\begin{proof}
According to the Proposition~\ref{取り得る値　２次元}, the Kirkwood-Dirac distribution is expressed as a linear combination of delta functions
\begin{equation}
\label{デルタ関数の前の係数　２次元}
K_{\rho}(x,y)=\sum_{(\sigma,\tau)=(\pm,\pm)} K_{\sigma\tau}\,\delta\qty(x\pm\frac{\sigma}{2})\delta\qty(y\pm\frac{\tau}{2}),
\end{equation}
where \(K_{\pm\pm}\in\mathbb{C}\) are the coefficients specified below.
Since the density operator of any state of the two-state quantum system can be expressed as
\begin{equation}
\label{density operator 2D}
\rho=
\begin{pmatrix}
a & b-ic\\
b+ic & 1-a
\end{pmatrix}
\end{equation}
with \(a,b,c\in\mathbb{R}\), the coefficients \(K_{\pm\pm}\) are given by
\begin{equation}
\label{K++}
K_{++} = \frac{1-i}{4} + \frac{i}{2}a + \frac{1}{2}(b+c),
\end{equation}
\begin{equation}
\label{K+-}
K_{+-} = \frac{1+i}{4} - \frac{i}{2}a + \frac{1}{2}(b-c),
\end{equation}
\begin{equation}
\label{K-+}
K_{-+} = \frac{1+i}{4} - \frac{i}{2}a - \frac{1}{2}(b-c),
\end{equation}
\begin{equation}
\label{K--}
K_{--} = \frac{1-i}{4} + \frac{i}{2}a - \frac{1}{2}(b+c).
\end{equation}
Arranging this relation conversely, we obtain
\begin{equation}
a = \frac{1+i}{2}-iK_{++}-iK_{--},
\end{equation}
\begin{equation}
b = -\frac{1}{2}+K_{++}+K_{+-},
\end{equation}
\begin{equation}
c = -\frac{1}{2}+K_{++}+K_{-+}.
\end{equation}
In other words, \((a,b,c)\) are uniquely determined by \(\{K_{\pm\pm}\}\).
Therefore, the density operator \(\rho\), and thus the state of the system, can be completely distinguished by the Kirkwood-Dirac distribution.
\end{proof}
This proposition also indicates that the Kirkwood-Dirac distribution is a good quasi-joint probability distribution.
Examples in Subsec.~\ref{Examples}, such as the ones in Fig.~\ref{z+-状態のS関数　図}, also showed that the state of a two-state quantum system cannot necessarily be distinguished by a quasi-joint probability distribution, but Prop.~\ref{状態を区別できる　２次元スピン} implies that it can be done by the Kirkwood-Dirac distribution.

The third question is for what states the quasi-joint-probability distribution takes complex values.
In Subsec.~\ref{Quasi-Classicalizations generating Real quasi-probabilities}, we have seen that some quasi-joint-probability distributions are real with respect to arbitrary states and others are not and that the Kirkwood-Dirac distribution belongs to the latter.
Similarly to the fact that negativity of the Wigner distribution of the position and the momentum of a particle tells the quantumness of a state, the complexity of the Kirkwood-Dirac distribution of the \(x\) and \(y\) components of the spin \(1/2\) gives us information about the state of the two-state quantum system.
\begin{prop}
\label{xy平面にいる　２次元}
The following two statements are equivalent to each other:
(i) a state of the two-state quantum system is in the \(xy\)-plane in the Bloch space;
(ii) the Kirkwood-Dirac distribution with respect to the state takes real values:
\begin{equation}
\langle J_3 \rangle _\rho \coloneqq \mathrm{Tr}(J_3\rho) = 0. 
\iff ^\forall (x,y) \in \mathbb{R}^2, K_{\rho}(x,y) \in \mathbb{R}.
\end{equation}
\end{prop}
To prove this proposition, we first prove the following lemma:
\begin{lem}
\label{xy面内の状態の密度行列}
The density operator of the state that is in the \(xy\)-plane in the Bloch space is expressed in the form
\begin{equation}
\label{density operator of the state in xy-plane}
\rho= \frac{1}{2}
\begin{pmatrix}
1 & e^{-i\phi}\sin\theta\\
e^{i\phi}\sin\theta & 1
\end{pmatrix}
\end{equation}
by using two parameters \(\theta\in[0,\pi]\) and \(\phi\in[0,2\pi)\).
\end{lem}
\begin{proof}
First, the pure state \(\ket{\psi(\theta,\phi)}\), a state which is described as \((\langle J_1\rangle,\,\langle J_2\rangle,\,\langle J_3\rangle)=1/2\,(\sin\theta\cos\phi,\,\sin\theta\sin\phi,\,\cos\theta)\), is expressed as a vector
\begin{equation}
\ket{\psi(\theta,\phi)} =
\begin{pmatrix}
\cos\frac{\theta}{2}\\
e^{i\phi}\sin\frac{\theta}{2}
\end{pmatrix},
\end{equation}
and then as a density operator
\begin{equation}
\label{純粋状態の密度演算子}
\ket{\psi(\theta,\phi)}\bra{\psi(\theta,\phi)} = \frac{1}{2}
\begin{pmatrix}
1+\cos\theta & e^{-i\phi}\sin{\theta}\\
e^{i\phi}\sin\theta & 1-\cos\theta
\end{pmatrix}.
\end{equation}
Next, we will consider the way to express the mixed state by the mixture of two pure states.
Let \(\rho\) denote the point in the Bloch space expressing the state in question, let \(l\) denote the straight line perpendicular to the \(xy\)-plane passing through \(\rho\), and let \(\rho_1\) and \(\rho_2\) denote the intersections of \(l\) and the Bloch sphere.
We then describe the polar coordinate of \(\rho_1\) as \((1/2,\,\theta,\,\phi)\) and that of \(\rho_2\) as \((1/2,\,\pi-\theta,\,\phi)\); see Fig.~\ref{Bloch球内の状態　図}.
\begin{figure}[H]
\centering
\includegraphics[width=0.4\textwidth]{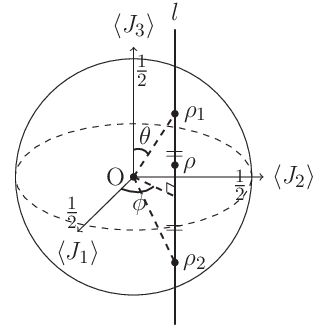}
\caption{Location of a mixed state \(\rho\) and two pure states \(\rho_1\) and \(\rho_2\) in the Bloch space.}
\label{Bloch球内の状態　図}
\end{figure}
Because of the definition of \(\rho_1\) and \(\rho_2\), \(\rho\) is always an internally dividing point of \(\rho_1\) and \(\rho_2\):
\begin{equation}
\langle J_i\rangle_\mathrm{\rho} = m\langle J_i\rangle_\mathrm{\rho_1} + (1-m)\langle J_i\rangle_\mathrm{\rho_2},\;\;\;m\in[0,1],\;i=1,2,3.
\end{equation}
Here, \(\langle J_i\rangle_\mathrm{\chi}\;(\chi=\rho_1,\rho_2,\rho_3)\) denote the expectation value of \(J_i\) of the state expressed by the point \(chi\) in the Bloch space.
In particular, \(\rho\) becomes the middle point of \(\rho_1\) and \(\rho_2\) iff \(\rho\) is in the \(xy\)-plane:
\begin{equation}
\langle J_i\rangle_\rho = \frac{1}{2}\langle J_i\rangle_{\rho_1} + \frac{1}{2}\langle J_i\rangle_{\rho_2}\;(i=1,2,3) \iff \rho \mbox{ is in the \(xy\)-plane}.
\end{equation}
Because of the linearity of expectations in quantum theory and the fact that the state of the two-state quantum system is identified by a combination of the expectation values of spins in the three directions, the density operator of the state expressed by the point \(\rho\) in the Bloch space becomes
\begin{align*}
\label{mixed stateのdensity operator 2D}
\rho_\mathrm{\rho} 
&= m\rho_1 +(1-m)\rho_2\\
&= m\ket{\psi(\theta,\phi)}\bra{\psi(\theta,\phi)} + (1-m)\ket{\psi(\pi-\theta,\phi)}\bra{\psi(\pi-\theta,\phi)}\\
&=\frac{1}{2}
\begin{pmatrix}
1+(2m-1)\cos\theta & e^{-i\phi}\sin{\theta}\\
e^{i\phi}\sin\theta & 1-(2m-1)\cos\theta
\end{pmatrix}.
\stepcounter{equation}\tag{\theequation}
\end{align*}
By substituting \(m=1/2\) to Eq.~(\ref{mixed stateのdensity operator 2D}), we obtain Eq.~(\ref{density operator of the state in xy-plane}).
\end{proof}
By using this lemma, Prop.~\ref{xy平面にいる　２次元} is easily proved:
\begin{proof}[Proof of Proposition \ref{xy平面にいる　２次元}]
By looking at Eqs.~(\ref{デルタ関数の前の係数　２次元}), (\ref{K++})--(\ref{K--}), we find 
that the Kirkwood-Dirac distribution \(K_{\rho}(x,y)\) takes real values at all points \((x,y)\) iff \(a = 1/2\). 
Then by substituting the density operator (\ref{mixed stateのdensity operator 2D}) to (\ref{density operator 2D}), we see that
\begin{equation}
a=\frac{1}{2} \iff m=\frac{1}{2} \iff \mbox{\(\rho\) is in the \(xy\)-plane}.
\end{equation}
Therefore
\begin{align*}
^\forall (x,y) \in \mathbb{R}^2, K_{\rho}(x,y) \in \mathbb{R}
&\iff a = 1/2\\
&\iff m = 1/2\\
&\iff \langle J_3 \rangle _\rho = 0.
\stepcounter{equation}\tag{\theequation}
\end{align*}
\end{proof}
Although the value of a quasi-joint-probability distribution is not necessarily positive semi-definite in contrast to the standard probabilities, it may be argued that it is more natural to be real than to take complex values from the perspective of interpreting the quasi-joint-probability as a probability.
Proposition \ref{xy平面にいる　２次元} can be interpreted that the Kirkwood-Dirac distribution of the \(x\) and \(y\) components of the spin \(1/2\) becomes natural in the sense that it takes real values when the state lives in the \(xy\)-plane.
In addition, we can easily see whether the expectation value of the spin of the state has non-zero components in a specific direction by looking at the imaginary part of the Kirkwood-Dirac distribution of two directions of spin orthogonal to the direction in question.

\subsection{Other quasi-classicalizations}
\label{Other Quasi-Classicalizations}
In Subsec.~\ref{General Features of Kirkwood-Dirac Distributions of Spin 1/2}, we understood that the Kirkwood-Dirac distribution is a quasi-joint-probability distribution for the two-state quantum system that behaves nicely.
In this subsection, we will see the conditions for the way of quasi-classicalization so that the corresponding quasi-joint-probability distributions may satisfy such properties.
In this subsection, we restrict the hashed operators to the unitary ones, which can be written in the forms
\begin{equation}
\label{＃T1}
\hat{\#}^{T1}_{J_1,J_2}(s,t) = e^{-isa_1J_1}e^{-itb_1J_2}e^{-isa_2J_1}e^{-itb_2J_2}...e^{-isa_nJ_1}e^{-itb_nJ_2},
\end{equation}
\begin{equation}
\label{＃T2}
\hat{\#}^{T2}_{J_1,J_2}(s,t) = e^{-isa_1J_1}e^{-itb_1J_2}e^{-isa_2J_1}e^{-itb_2J_2}...e^{-isb_nJ_2}e^{-ita_{n+1}J_1},
\end{equation}
\begin{equation}
\label{＃T3}
\hat{\#}^{T3}_{J_1,J_2}(s,t) = e^{-isb_1J_2}e^{-ita_1J_1}e^{-isb_2J_2}e^{-ita_2J_1}...e^{-isb_nJ_2}e^{-ita_nJ_1},
\end{equation}
\begin{equation}
\label{＃T4}
\hat{\#}^{T4}_{J_1,J_2}(s,t) = e^{-isb_1J_2}e^{-ita_1J_1}e^{-isb_2J_2}e^{-ita_2J_1}...e^{-isa_nJ_1}e^{-itb_{n+1}J_2}
\end{equation}
for \(\sum_ia_i=\sum_ib_i=1,\,n\in\mathbb{N}\). 
These include \(\hat{\#}^W_{J_1,J_2}(s,t)\) in the \(n \to \infty\) limit.
This type of hashed operator becomes an element of the fundamental representation of \(SU(2)\).
Though hashed operators are permitted to take a more general form, unitary ones are the basic form and others are understood as an extension of this form.

First, we will look at the restrictions so that the quasi-joint-probability may take non-zero values only at the possible values of the observables in question.
The conditions for a quasi-joint-probability distribution to take non-zero values only at the possible values of the \(x\) and \(y\) components of the spin $1/2$ is that the components of the quasi-joint-spectrum distribution become a linear combination of \(\delta(x\pm1/2)\delta(y\pm1/2)\), or equivalently, that the components of the hashed operator to become a linear combination of \(e^{\pm is/2\pm it/2}\).
Since the components of 
\begin{equation}
\label{J1の指数関数}
e^{-isa_iJ_1} = \frac{1}{2}
\begin{pmatrix}
e^{-i\frac{s}{2}a_i}+e^{i\frac{s}{2}a_i} & e^{-i\frac{s}{2}a_i}-e^{i\frac{s}{2}a_i}\\*
e^{-i\frac{s}{2}a_i}-e^{i\frac{s}{2}a_i} & e^{-i\frac{s}{2}a_i}+e^{i\frac{s}{2}a_i}
\end{pmatrix}
\end{equation}
include both \(e^{-ia_is/2}\) and \(e^{ia_is/2}\) and the components of
\begin{equation}
\label{J2の指数関数}
e^{-itb_iJ_2} = \frac{1}{2}
\begin{pmatrix}
e^{-i\frac{t}{2}b_i}+e^{i\frac{t}{2}b_i} & -i(e^{-i\frac{t}{2}b_i}-e^{i\frac{t}{2}b_i})\\
i(e^{-i\frac{t}{2}b_i}-e^{i\frac{t}{2}b_i}) & e^{-i\frac{t}{2}b_i}+e^{i\frac{t}{2}b_i}
\end{pmatrix}
\end{equation}
include both \(e^{-ib_it/2}\) and \(e^{ib_it/2}\), the components of the hashed operators (\ref{＃T1})--(\ref{＃T4}) generally include the term \(e^{-is/2(\pm a_1\pm a_2...\pm a_n)-it/2(\pm b_1 \pm b_2...\pm b_n)}\).
Therefore, we can say that the Kirkwood-Dirac type \(e^{-isJ_1}e^{-itJ_2}\) and \(e^{-itJ_2}e^{-isJ_1}\) are the only choice of hashed operator that generates a quasi-joint-probability distribution which takes non-zero values only at the possible values of the \(x\) and \(y\) components of the spin \(1/2\).

Second, we will consider the restrictions to the choice of the quasi-joint-spectrum distribution so that the state of the two-state quantum system can be distinguished by the quasi-joint-probability distribution generated from the quasi-joint-spectrum distribution.
The answer to this question is shown in the following proposition:
\begin{prop}
\label{擬確率分布で状態を区別できるための条件 2D}
The following two statements are equivalent to each other: (i) the states of a two-state quantum system can be distinguished by the quasi-joint-probability distribution of the \(x\) and \(y\) components of the spin \(1/2\); (ii) the corresponding quasi-joint-spectrum distribution is non-Hermitian.
\end{prop}
To prove this proposition, we first prove the following lemmas:
\begin{lem}
\label{z+とz-を区別できればいい}
A necessary and sufficient condition for a quasi-classicalization so that the state of a two-state quantum system may be distinguishable by the quasi-joint-probability distribution is given by the following relation of components of the quasi-joint-spectrum distribution:
\begin{equation}
\bigl(\#_{J_1,J_2}(x,y)\bigr)_{11} \neq \bigl(\#_{J_1,J_2}(x,y)\bigr)_{22}.
\end{equation}
\end{lem}
\begin{proof}
Since any quasi-joint-probability distributions give the same marginal probability distributions, any quasi-joint-probability distributions of the \(x\) and \(y\) components of the spin \(1/2\) can distinguish two states whose expectation values of \(J_1\) or \(J_2\) differ from each other.
Therefore, we can distinguish all the states of a two-state quantum system by a quasi-joint-probability distribution of the \(x\) and \(y\) components of the spin \(1/2\) if we can distinguish two states that differ only its \(z\)-coordinates in the Bloch space.
Because of the linearity of the expectations in the quantum theory, we thus conclude that it is sufficient to distinguish \(\ket{\psi(\theta,\phi)}\) and \(\ket{\psi(\pi-\theta,\phi)}\).

Since the quasi-joint-probability distribution with respect to \(\ket{\psi(\theta,\phi)}\) and \(\ket{\psi(\pi-\theta,\phi)}\) are given by
\begin{align*}
P_{\theta,\phi}(x,y) = &\bigl(\hat{\#}_{J_1,J_2}(x,y)\bigr)_{11}\frac{1+\cos\theta}{2} +\bigl(\hat{\#}_{J_1,J_2}(x,y)\bigr)_{12}\frac{e^{i\phi}\sin\theta}{2}\\
&+\bigl(\hat{\#}_{J_1,J_2}(x,y)\bigr)_{21}\frac{e^{-i\phi}\sin\theta}{2} +\bigl(\hat{\#}_{J_1,J_2}(x,y)\bigr)_{22}\frac{1-\cos\theta}{2}
\stepcounter{equation}\tag{\theequation}
\end{align*}
and
\begin{align*}
P_{\pi-\theta,\phi}(x,y) = &\bigl(\hat{\#}_{J_1,J_2}(x,y)\bigr)_{11}\frac{1-\cos\theta}{2} +\bigl(\hat{\#}_{J_1,J_2}(x,y)\bigr)_{12}\frac{e^{i\phi}\sin\theta}{2}\\
&+\bigl(\hat{\#}_{J_1,J_2}(x,y)\bigr)_{21}\frac{e^{-i\phi}\sin\theta}{2} +\bigl(\hat{\#}_{J_1,J_2}(x,y)\bigr)_{22}\frac{1+\cos\theta}{2},
\stepcounter{equation}\tag{\theequation}
\end{align*}
respectively, \(P_{\theta,\phi}(x,y)=P_{\pi-\theta,\phi}(x,y)\) is equivalent to \(\bigl(\hat{\#}_{J_1,J_2}(x,y)\bigr)_{11} = \bigl(\hat{\#}_{J_1,J_2}(x,y)\bigr)_{22}\).
Therefore \(\bigl(\hat{\#}_{J_1,J_2}(x,y)\bigr)_{11} \neq \bigl(\hat{\#}_{J_1,J_2}(x,y)\bigr)_{22}\) is a necessary and sufficient condition for the corresponding quasi-joint-probability distribution to be able to distinguish the states.
\end{proof}
This lemma also indicates that every state of the two-state quantum system is distinguishable by a quasi-joint-probability distribution of the \(x\) and \(y\) components of the spin \(1/2\) if the \(\ket{z+}\) state and the \(\ket{z-}\) state is distinguishable by the quasi-joint-probability distribution since \(P_{z+}(x,y) = \bigl(\#_{J_1,J_2}(x,y)\bigr)_{11},\;P_{z-}(x,y) = \bigl(\#_{J_1,J_2}(x,y)\bigr)_{22}\).
\begin{lem}
\label{ユニタリ行列の性質}
When considering an element of the fundamental representation of \(SU(2)\)
\begin{equation}
U =
\begin{pmatrix}
\mu & \nu\\
\kappa & \lambda
\end{pmatrix}
\in SU(2),
\end{equation}
\(\mu=\lambda\) is equivalent to \(\mu,\lambda\in\mathbb{R}\).
\end{lem}
\begin{proof}
Since \(U \in SU(2)\),
\begin{equation}
\det U = \mu\lambda-\nu\kappa =1,
\end{equation}
\begin{equation}
\label{unitarity 2D}
UU^{\dagger} =
\begin{pmatrix}
\lvert\mu\rvert^2+\lvert\nu\rvert^2 & \mu\kappa^*+\nu\lambda^*\\
\mu^*\kappa+\nu^*\lambda & \lvert\kappa\rvert^2+\lvert\lambda\rvert^2
\end{pmatrix}
=
\begin{pmatrix}
1 & 0\\
0 & 1
\end{pmatrix},
\end{equation}
\begin{equation}
U^{\dagger}U =
\begin{pmatrix}
\lvert\mu\rvert^2+\lvert\kappa\rvert^2 & \mu^*\nu+\kappa^*\lambda\\
\mu\nu^*+\kappa\lambda^* & \lvert\nu\rvert^2+\lvert\lambda\rvert^2
\end{pmatrix}
=
\begin{pmatrix}
1 & 0\\
0 & 1
\end{pmatrix}
\end{equation}
hold.
Therefore, 
\begin{align*}
\label{SU(2) 恒等式}
\lvert\mu\lambda-1\rvert^2 &= \lvert\nu\rvert^2\lvert\kappa\rvert^2 = (1-\lvert\mu\rvert^2)(1-\lvert\lambda\rvert^2)\\
2\mathrm{Re}(\mu\lambda) &= \lvert\mu\rvert^2 + \lvert\lambda\rvert^2.
\stepcounter{equation}\tag{\theequation}
\end{align*}
Hence, if \(\mu,\lambda\in\mathbb{R}\), Eq.~(\ref{SU(2) 恒等式}) becomes
\begin{align}
2\mu\lambda = \mu^2+\lambda^2,
\end{align}
thereby pointry to \(\mu=\nu\).
Conversely if \(\mu=\lambda\), Eq~(\ref{SU(2) 恒等式}) becomes
\begin{equation}
\mathrm{Re}(\mu^2) = \lvert\mu\rvert^2,
\end{equation}
thereby pointry to \(\mu=\nu\in\mathbb{R}\).
\end{proof}
In addition, since 
\begin{equation}
\mu\kappa^*+\nu\lambda = 0
\end{equation}
always holds because of the unitarity (\ref{unitarity 2D}),
\begin{equation}
\kappa^* = -\nu
\end{equation}
automatically holds if \(0\neq\mu=\lambda\in\mathbb{R}\).
\begin{lem}
\label{hahshed operator の変換性}
Under the transformation \((s,t)\to(-s,-t)\), the diagonal components of a hashed operator of the \(x\) and \(y\) components of the spin for two-state quantum system are invariant and the non-diagonal components change their signs.
\end{lem}
\begin{proof}
First, we point out the fact that if the two-dimensional square matrices 
\begin{equation}
A(s,t)=
\begin{pmatrix}
A_{11}(s,t) & A_{12}(s,t)\\
A_{21}(s,t) & A_{22}(s,t)
\end{pmatrix},\;\;\;
B = 
\begin{pmatrix}
B_{11}(s,t) & B_{12}(s,t)\\
B_{21}(s,t) & B_{22}(s,t)
\end{pmatrix}
\end{equation}
satisfy the property that their diagonal components are invariant and non-daigonal components change their signs as in
\begin{equation}
\label{パラメタ反転への変換性}
A(-s,-t)=
\begin{pmatrix}
A_{11}(-s,-t) & A_{12}(-s,-t)\\
A_{21}(-s,-t) & A_{22}(-s,-t)
\end{pmatrix}
=
\begin{pmatrix}
A_{11}(s,t) & -A_{12}(s,t)\\
-A_{21}(s,t) & A_{22}(s,t)
\end{pmatrix},
\end{equation}
\begin{equation}
B(-s,-t)=
\begin{pmatrix}
B_{11}(-s,-t) & B_{12}(-s,-t)\\
B_{21}(-s,-t) & B_{22}(-s,-t)
\end{pmatrix}
=
\begin{pmatrix}
B_{11}(s,t) & -B_{12}(s,t)\\
-B_{21}(s,t) & B_{22}(s,t)
\end{pmatrix},
\end{equation}
under the said transformation, the product of A and B
\begin{equation}
AB(s,t)=
\begin{pmatrix}
A_{11}(s,t)B_{11}(s,t)+A_{12}(s,t)B_{21}(s,t) & A_{11}(s,t)B_{12}(s,t)+A_{12}(s,t)B_{22}(s,t)\\
A_{21}(s,t)B_{11}(s,t)+A_{22}(s,t)B_{12}(s,t) & A_{21}(s,t)B_{12}(s,t)+A_{22}(s,t)B_{22}(s,t)
\end{pmatrix}
\end{equation}
retains the same properties.
Then since (\ref{J1の指数関数}) and (\ref{J2の指数関数}) transform as the same as (\ref{パラメタ反転への変換性}), the hashed operators (\ref{＃T1})--(\ref{＃T4}) and also transform in the same way.
\end{proof}
Now we prove Prop.~\ref{擬確率分布で状態を区別できるための条件 2D} by using the above lemmas:
\begin{proof}[Proof of Proposition \ref{擬確率分布で状態を区別できるための条件 2D}]
From Lemma \ref{z+とz-を区別できればいい}, we see that the state of a two-state quantum system cannot be distinguished by the quasi-joint-probability distribution of the \(x\) and \(y\) components of the spin \(1/2\) if and only if \(\bigl(\#_{J_1,J_2}(x,y)\bigr)_{11} = \bigl(\#_{J_1,J_2}(x,y)\bigr)_{22}\), in terms of their Fourier transforms, which is equivalent to \(\bigl(\hat{\#}_{J_1,J_2}(s,t)\bigr)_{11} = \bigl(\hat{\#}_{J_1,J_2}(s,t)\bigr)_{22}\).

From Lemma \ref{ユニタリ行列の性質}, which is equivalent to \(\bigl(\hat{\#}_{J_1,J_2}(s,t)\bigr)_{11}\in\mathbb{R},\;\bigl(\hat{\#}_{J_1,J_2}(s,t)\bigr)_{22}\in\mathbb{R}\).
Here, \(\bigl(\hat{\#}_{J_1,J_2}(s,t)\bigr)_{21}^* = -\bigl(\hat{\#}_{J_1,J_2}(s,t)\bigr)_{12}\) also holds automatically since the normalization condition for the quasi-joint-probability distribution with respect to the states \(\ket{z\pm}\) requires \(\bigl(\hat{\#}_{J_1,J_2}(s,t)\bigr)_{11}\neq0\) and \(\bigl(\hat{\#}_{J_1,J_2}(s,t)\bigr)_{22}\neq0\).
Then since \(\bigl(\hat{\#}_{J_1,J_2}(-s,-t)\bigr)_{12} = -\bigl(\hat{\#}_{J_1,J_2}(s,t)\bigr)_{12}\) from Lemma \ref{hahshed operator の変換性}, one finds \(\bigl(\hat{\#}_{J_1,J_2}(s,t)\bigr)_{21}^* = \bigl(\hat{\#}_{J_1,J_2}(-s,-t)\bigr)_{12}\) if \(\bigl(\hat{\#}_{J_1,J_2}(s,t)\bigr)_{11} = \bigl(\hat{\#}_{J_1,J_2}(s,t)\bigr)_{22}\in\mathbb{R}\).

We therefore conclude that the distinguishability is equivalent to \(\hat{\#}_{J_1,J_2}^\dagger(-s,-t) = \hat{\#}_{J_1,J_2}(s,t)\), which according to (\ref{real hashed operator}), in turn is equivalent to the Hermiticity of the quasi-joint-spectrum distribution. 
\end{proof}
This proposition is quite implicative. From Subsec.~\ref{Quasi-Classicalizations generating Real quasi-probabilities} we know that Hermitian quasi-joint-spectrum distributions quasi-classicalize every state to real-valued quasi-joint-probability distribution and quantize every physical quantity to a Hermitian operator.
Therefore, Prop.~\ref{擬確率分布で状態を区別できるための条件 2D} implies that in the two-state quantum system, in order to distinguish the states by a quasi-joint-probability distribution of the \(x\) and \(y\) components of the spin \(1/2\), we must use a quasi-joint-spectrum distributions that quasi-classicalizes a state to a complex-valued quasi-joint-probability distribution and quantizes a physical quantity to a non-Hermitian matrix.
We need the imaginary part of a quasi-joint-probability distribution to distinguish the states.

\subsection{Other observables}
\label{Other observables in 2D}
In the previous subsections, we have investigated the features of quasi-joint-probability distributions of the \(x\) and \(y\) components of the spin \(1/2\) for the two-state quantum system.
In this subsection, we will turn our attention to the quasi-joint-probability distributions of other observables.
The following proposition for the Kirkwood-Dirac distribution holds.
\begin{prop}
\label{他の物理量についてのKirkwood-Dirac関数　2D}
The state of the two-state quantum system can be distinguished by the Kirkwood-Dirac distribution of any pair of non-commutative observables.
\end{prop}
\begin{proof}
Two-dimensional Hermitian matrices, i.e., observables of the two-state quantum system, can be expressed as a linear combination of the identity operator and the elements of the fundamental representation of \(\mathfrak{su}(2)\):
\begin{equation}
A = a_0I + \sum_{i=1}^{3}a_i\sigma_i,\;\;\;a_i\in\mathbb{R}\;\;\mbox{for}\; i=0,1,2,3.
\end{equation}

First, we will see that it is sufficient to prove the statement about observables of limited form. 
Scalar multiplication of the observable is irrelevant to whether the states of a two-state quantum system can be distinguished by the Kirkwood-Dirac distribution, because it only expands the distribution:
\begin{align*}
K^{\alpha A,B}_\rho(x,y)
&= \frac{1}{2\pi}\mathrm{Tr}[\rho\mathcal{F}^{-1}\bigl(\hat{\#}^K_{\alpha A,B}(s,t)\bigr)(x,y)]\\
&= \frac{1}{2\pi}\mathrm{Tr}[\rho\mathcal{F}^{-1}\bigl(e^{-is\alpha A}e^{-itB}\bigr)(x,y)]\\
&= \frac{1}{2\pi}\mathrm{Tr}[\rho\mathcal{F}^{-1}\bigl(\hat{\#}^K_{A,B}(\alpha s,t)\bigr)(x,y)]\\
&=K^{A,B}_\rho(x/\alpha,y).
\stepcounter{equation}\tag{\theequation}
\end{align*}
It is also irrelevant to add the identity operator to the observable concerned for the distinguishability of the state because it only translates the distribution:
\begin{align*}
K^{A+I,B}_\rho(x,y)
&= \mathrm{Tr}[\rho\#^K_{A+I,B}(x,y)]\\
&= \frac{1}{2\pi}\mathrm{Tr}[\rho\mathcal{F}^{-1}\bigl(e^{-is(A+I)}e^{-itB}\bigr)(x,y)]\\
&= \frac{1}{2\pi}\mathrm{Tr}[\rho\mathcal{F}^{-1}\bigl(e^{-is}e^{-isA}e^{-itB}\bigr)(x,y)]\\
&= \mathrm{Tr}[\rho\#^K_{A,B}(x-1,y)]\\
&=K^{A,B}_\rho(x-1,y).
\stepcounter{equation}\tag{\theequation}
\end{align*}
Therefore, we only have to consider the observables corresponding to the elements of \(\mathfrak{su}(2)\).
Moreover, since the quasi-joint-probability distribution are invariant under unitary transformations
\begin{equation}
\rho\to U\rho U^\dagger,\;\;\;\hat{\#}\to U\hat{\#}U^\dagger,
\end{equation}
we conclude that it is sufficient to prove the statement for the pair of non-commutative observables
\begin{equation}
A = \sigma_3,\;\;\;\hat{B} = \sigma_3\cos\varphi+\sigma_1\sin\varphi,\;\;\;\varphi\neq n\pi,\;n\in\mathbb{Z}.
\end{equation}
Since
\begin{equation}
e^{-itB}=\frac{1}{2}
\begin{pmatrix}
(1+\cos\varphi)e^{-it}+(1-\cos\varphi)e^{it} & \sin\varphi(e^{-it}-e^{it})\\
\sin\varphi(e^{-it}-e^{it}) & (1-\cos\varphi)e^{-it}+(1+\cos\varphi)e^{it}
\end{pmatrix},
\end{equation}
the hashed operator reads
\begin{equation}
\hat{\#}^{K}_{A,B}(s,t)=\frac{1}{2}
\begin{pmatrix}
e^{-is}\{(1+\cos\varphi)e^{-it}+(1-\cos\varphi)e^{it}\} & e^{-is}(e^{-it}-e^{it})\sin\varphi\\
e^{is}(e^{-it}-e^{it})\sin\varphi & e^{is}\{(1-\cos\varphi)e^{-it}+(1+\cos\varphi)e^{it}\}
\end{pmatrix}.
\end{equation}
Therefore, the Kirkwood-Dirac distribution of \(A\) and \(B\) with respect to  the state (\ref{density operator 2D}) becomes
\begin{equation}
K^{A,B}_{\rho}(x,y)=\sum_{(\pm,\pm)} K^{\varphi}_{\pm\pm}\,\delta\qty(x\pm\frac{1}{2})\delta\qty(y\pm\frac{1}{2}),
\end{equation}
\begin{equation}
K^{\varphi}_{++}=\frac{1+\cos\varphi}{2}a+\frac{\sin\varphi}{2}(b+ic),
\end{equation}
\begin{equation}
K^{\varphi}_{+-}=\frac{1-\cos\varphi}{2}a-\frac{\sin\varphi}{2}(b+ic),
\end{equation}
\begin{equation}
K^{\varphi}_{-+}=\frac{1-\cos\varphi}{2}(1-a)+\frac{\sin\varphi}{2}(b-ic),
\end{equation}
\begin{equation}
K^{\varphi}_{--}=\frac{1+\cos\varphi}{2}(1-a)-\frac{\sin\varphi}{2}(b-ic).
\end{equation}
Conversely,
\begin{equation}
a=K^{\varphi}_{++}+K^{\varphi}_{+-},
\end{equation}
\begin{equation}
b=\frac{(K^{\varphi}_{++}+K^{\varphi}_{-+}-1)}{\sin\varphi}-\frac{K^{\varphi}_{++}+K^{\varphi}_{+-}}{\tan\varphi},
\end{equation}
\begin{equation}
c=\frac{1}{i\sin\varphi}\Bigl(K^{\varphi}_{++}+K^{\varphi}_{--}-\frac{1+\cos\varphi}{2}\Bigr),
\end{equation}
are determined from \(\{K^{\varphi}_{\pm\pm}\}\), which points to the distinguishability of the state by the Kirkwood-Dirac distribution.
\end{proof}
From this proposition, we see that the useful properties explored before are not only limited to the \(x\) and \(y\) components of the spin \(1/2\), which further points to the usefulness of the Kirkwood-Driac distribution in analyzing the two-state quantum systems.

\section{Quasi-Probabilities for Three-State Quantum System}
\label{Quasi-Probabilities for Three-State Quantum System}
In Sec.~\ref{Quasi-Probabilities for Two-State Quantum System}, we have explored the features of the quasi-joint-probability distribution, especially the Kirkwood-Dirac distribution, for the two-state quantum system.
From now on we will turn our attention to higher-dimensional systems and look into the features of the quasi-probabilities.
In this section, we will focus on the three-state quantum system, which is the simplest system after the two-state quantum system.
Though the three-state quantum system is much more complex than the two-state system, we will see that the Kirkwood-Dirac distribution still enjoys similar properties to the two-dimensional case.

\subsection{Independ components of a state}
\label{Independ Components of a State}
In this subsection, we will see the peculiarity of the two-state quantum system compared to higher-dimensional quantum system in analyzing it by two directions of spin and see the difference between the three-state and two-state quantum systems.

In the two-state quantum system, a state of the system corresponds one-to-one to a point in the Bloch sphere.
However, more complex systems cannot be determined only by the expectation values of three components of the spins.
In general, we need at least \(N^2-1\) observables in order to distinguish the state of the \(N\)-state quantum system by their expectation values.
In contrast to the expectation values, the Kirkwood-Dirac distributions only need two directions of spins in order to distinguish the states of the two-state quantum system.
Then, the question here is how many observables we need to distinguish the state of the three-state quantum system by the quasi-joint-probability distributions and what information about the system we can obtain from the quasi-joint-probability distributions of two directions of spins.
We will see the answers to these questions in the following subsections.

\subsection{General feature of the Kirkwood-Dirac distributions of spin 1}
\label{General Feature of Kirkwood-Dirac Distributions of Spin 1}
The spin operator for the three-state quantum system is the three-dimensional irreducible representation, which is called the spin 1 representation of \(\mathfrak{su}(2)\):
\begin{equation}
J_1^{(1)} = \frac{1}{\sqrt{2}}
\begin{pmatrix}
0 & 1 & 0\\
1 & 0 & 1\\
0 & 1 & 0
\end{pmatrix},\;\;
J_2^{(1)} = \frac{1}{\sqrt{2}}
\begin{pmatrix}
0 & -i & 0\\
i & 0 & -i\\
0 & i & 0
\end{pmatrix},\;\;
J_3^{(1)} =
\begin{pmatrix}
1 & 0 & 0\\
0 & 0 & 0\\
0 & 0 & -1
\end{pmatrix}.
\end{equation}
The value they can take is restricted to the eigenvalues of spin \(1\), namely, \(0\) and \(\pm1\).
Similarly to the two-dimensional case, the Kirkwood-Dirac distribution of the \(x\) and \(y\) components of spin 1 for the three-state quantum system enjoys the following properties.

First, we will look at the problem regarding the possible values of the observables.
\begin{prop}
\label{取り得る値　３次元}
The Kirkwood-Dirac distribution of the \(x\) and \(y\) components of spin 1 with respect to an arbitrary state of the three-state quantum system takes non-zero value only at \((J_1,J_2)=(\pm1,\pm1),\,(\pm1,0),\,(0,\pm1)\):
\begin{equation}
^\forall\rho,\; K_{\rho}(x,y) = 0\;\;\bigl(\mbox{for}\;(x,y)\neq(\pm1,\pm1),\,(\pm1,0),\,(0,\pm1)\bigl).
\end{equation}
\end{prop}
\begin{proof}
The exponential function of the spin \(x\) and \(y\) of spin 1 representation is
\begin{equation}
e^{-isJ_1^{(1)}}=\frac{1}{2}
\begin{pmatrix}
\frac{1}{2}(e^{-is}+e^{is})+1 & \frac{1}{\sqrt{2}}(e^{-is}-e^{is}) & \frac{1}{2}(e^{-is}+e^{is})-1\\
\frac{1}{\sqrt{2}}(e^{-is}-e^{is}) & e^{-is}+e^{is} & \frac{1}{\sqrt{2}}(e^{-is}-e^{is})\\
\frac{1}{2}(e^{-is}+e^{is})-1 & \frac{1}{\sqrt{2}}(e^{-is}-e^{is}) & \frac{1}{2}(e^{-is}+e^{is})-1\\
\end{pmatrix},
\end{equation}
\begin{equation}
e^{-itJ_2^{(1)}}=\frac{1}{4}
\begin{pmatrix}
\frac{1}{2}(e^{-it}+e^{it})+1 & \frac{-i}{\sqrt{2}}(e^{-it}-e^{it}) & \frac{-1}{2}(e^{-it}+e^{it})+1\\
\frac{i}{\sqrt{2}}(e^{-it}-e^{it}) & e^{-it}+e^{it} & \frac{-i}{\sqrt{2}}(e^{-it}-e^{it})\\
\frac{-1}{2}(e^{-it}+e^{it})+1 & \frac{i}{\sqrt{2}}(e^{-it}-e^{it}) & \frac{1}{2}(e^{-it}+e^{it})-1\\
\end{pmatrix}.
\end{equation}
The hashed operator for the Kirkwood-Dirac distribution of the \(x\) and \(y\) components of the spin 1 is defined in the same manner as the two-dimensional case:
\begin{equation}
\hat{\#}^{K}_{J_1^{(1)},J_2^{(1)}}(x,y)\coloneqq e^{-isJ_1^{(1)}}e^{-itJ_2^{(1)}}.
\end{equation}
Then the components of the hashed operator read
\begin{align}
\Bigl(\hat{\#}^{K}_{J_1^{(1)},J_2^{(1)}}(x,y)\Bigr)_{11}
&=\frac{i}{8}(e^{-is}-e^{is})(e^{-it}-e^{it})+\frac{1}{4}(e^{-is}+e^{is})+\frac{1}{4}(e^{-it}+e^{it}),\\
\Bigl(\hat{\#}^{K}_{J_1^{(1)},J_2^{(1)}}(x,y)\Bigr)_{12}
&=\frac{1}{4\sqrt{2}}(e^{-is}-e^{is})(e^{-it}+e^{it})-\frac{i}{2\sqrt{2}}(e^{-it}-e^{it}),\\
\Bigl(\hat{\#}^{K}_{J_1^{(1)},J_2^{(1)}}(x,y)\Bigr)_{13}
&=\frac{-i}{8}(e^{-is}-e^{is})(e^{-it}-e^{it})+\frac{1}{4}(e^{-is}+e^{is})-\frac{1}{4}(e^{-it}+e^{it}),\\
\Bigl(\hat{\#}^{K}_{J_1^{(1)},J_2^{(1)}}(x,y)\Bigr)_{21}
&=\frac{i}{4\sqrt{2}}(e^{-is}+e^{is})(e^{-it}-e^{it})+\frac{1}{2\sqrt{2}}(e^{-is}-e^{is}),\\
\Bigl(\hat{\#}^{K}_{J_1^{(1)},J_2^{(1)}}(x,y)\Bigr)_{22}
&=\frac{1}{4}(e^{-is}+e^{is})(e^{-it}+e^{it}),\\
\Bigl(\hat{\#}^{K}_{J_1^{(1)},J_2^{(1)}}(x,y)\Bigr)_{23}
&=\frac{-i}{4\sqrt{2}}(e^{-is}+e^{is})(e^{-it}-e^{it})+\frac{1}{2\sqrt{2}}(e^{-is}-e^{is}),\\
\Bigl(\hat{\#}^{K}_{J_1^{(1)},J_2^{(1)}}(x,y)\Bigr)_{31}
&=\frac{i}{8}(e^{-is}-e^{is})(e^{-it}-e^{it})+\frac{1}{4}(e^{-is}+e^{is})-\frac{1}{4}(e^{-it}+e^{it}),\\
\Bigl(\hat{\#}^{K}_{J_1^{(1)},J_2^{(1)}}(x,y)\Bigr)_{32}
&=\frac{1}{4\sqrt{2}}(e^{-is}-e^{is})(e^{-it}+e^{it})+\frac{i}{2\sqrt{2}}(e^{-it}-e^{it}),\\
\Bigl(\hat{\#}^{K}_{J_1^{(1)},J_2^{(1)}}(x,y)\Bigr)_{33}
&=\frac{-i}{8}(e^{-is}-e^{is})(e^{-it}-e^{it})+\frac{1}{4}(e^{-is}+e^{is})+\frac{1}{4}(e^{-it}+e^{it}).
\end{align}
On the other hand, the density operator of the three-state quantum system is generally expressed as
\begin{equation}
\hat{\rho}=
\begin{pmatrix}
a & b-ic & d-if\\
b+ic & 1-a-g & h-ik\\
d+if & h+ik & g
\end{pmatrix},
\end{equation}
for eight variables \(a,b,c,d,f,g,h,k\in\mathbb{R}\).
Therefore, the Kirkwood-Dirac distribution with respect to the state \(\hat{\rho}\) is calculated as
\begin{align}
\label{デルタ関数の線形結合　３次元}
K_{\rho}(x,y)&=\sum_{\alpha,\beta=-1,0,+1}K^{(1)}_{\alpha\beta}\delta(x-\alpha)\delta(x-\beta),\\
\label{K++ 3D}
K^{(1)}_{++}&=\frac{1}{4}-\frac{2-i}{8}a-\frac{2+i}{8}g+\frac{1+i}{4\sqrt{2}}(b+c)+\frac{1}{4}f+\frac{1-i}{4\sqrt{2}}(h+k),\\
\label{K+0 3D}
K^{(1)}_{+0}&=\frac{1}{4}(a+g)+\frac{1}{2\sqrt{2}}(b-ic)+\frac{1}{2}d+\frac{1}{2\sqrt{2}}(h+ik),\\
\label{K+- 3D}
K^{(1)}_{+-}&=\frac{1}{4}-\frac{2+i}{8}a-\frac{2-i}{8}g+\frac{1-i}{4\sqrt{2}}(b-c)-\frac{1}{4}f+\frac{1+i}{4\sqrt{2}}(h-k),\\
\label{K0+ 3D}
K^{(1)}_{0+}&=\frac{1}{4}(a+g)-\frac{i}{2\sqrt{2}}(b+ic)-\frac{1}{2}d+\frac{i}{2\sqrt{2}}(h-ik),\\
\label{K00 3D}
K^{(1)}_{00}&=0,\\
\label{K0- 3D}
K^{(1)}_{0-}&=\frac{1}{4}(a+g)+\frac{i}{2\sqrt{2}}(b+ic)-\frac{1}{2}d-\frac{i}{2\sqrt{2}}(h-ik),\\
\label{K-+ 3D}
K^{(1)}_{-+}&=\frac{1}{4}-\frac{2+i}{8}a-\frac{2-i}{8}g-\frac{1-i}{4\sqrt{2}}(b-c)-\frac{1}{4}f-\frac{1+i}{4\sqrt{2}}(h-k),\\
\label{K-0 3D}
K^{(1)}_{-0}&=\frac{1}{4}(a+g)-\frac{1}{2\sqrt{2}}(b-ic)+\frac{1}{2}d-\frac{1}{2\sqrt{2}}(h+ik),\\
\label{K-- 3D}
K^{(1)}_{--}&=\frac{1}{4}-\frac{2-i}{8}a-\frac{2+i}{8}g-\frac{1+i}{4\sqrt{2}}(b+c)+\frac{1}{4}f-\frac{1-i}{4\sqrt{2}}(h+k).
\end{align}
Since the Kirkwood-Dirac distribution is a linear combination of the delta functions as seen in Eq.~(\ref{デルタ関数の線形結合　３次元}), it takes non-zero values only at \((\pm1,\pm1),\,(\pm1,0),\,(0,\pm1)\).
\end{proof}
This proposition indicates that the property of the Kirkwood-Dirac distribution that it takes non-zero values only at the possible values of the observables concerned also holds in three-state quantum system. 

The second topic is what information of the three-state quantum system we can obtain from the quasi-joint-probability distribution.
\begin{prop}
\label{状態を区別できる　3D}
The state of a three-state quantum system is completely distinguishable by the Kirkwood-Dirac distribution of the \(x\) and \(y\) components of spin 1:
\begin{equation}
\rho=\rho' \iff ^\forall(x,y)\in\mathbb{R},\,K_{\rho}(x,y)=K_{\rho'}(x,y).
\end{equation}
\end{prop}
\begin{proof}
Relations (\ref{K++ 3D}) - (\ref{K-- 3D}) can be rewritten as
\begin{equation}
\label{Kirkwood-Dirac変換}
\bm{K^{(1)}} = T\bm{\rho},
\end{equation}
where
\begin{equation}
\label{Kirkwood-Dirac関数　行列表示}
\bm{K^{(1)}}=
\begin{pmatrix}
\mathrm{Re}(K^{(1)}_{++})-\frac{1}{4}\\
\mathrm{Im}(K^{(1)}_{++})\\
\mathrm{Re}(K^{(1)}_{+0})\\
\mathrm{Im}(K^{(1)}_{+0})\\
\mathrm{Re}(K^{(1)}_{+-})-\frac{1}{4}\\
\mathrm{Im}(K^{(1)}_{+-})\\
\mathrm{Re}(K^{(1)}_{0+})\\
\mathrm{Im}(K^{(1)}_{0+})\\
\mathrm{Re}(K^{(1)}_{0-})\\
\mathrm{Im}(K^{(1)}_{0-})\\
\mathrm{Re}(K^{(1)}_{-+})-\frac{1}{4}\\
\mathrm{Im}(K^{(1)}_{-+})\\
\mathrm{Re}(K^{(1)}_{-0})\\
\mathrm{Im}(K^{(1)}_{-0})\\
\mathrm{Re}(K^{(1)}_{--})-\frac{1}{4}\\
\mathrm{Im}(K^{(1)}_{--})\\
\end{pmatrix},\,
T=
\begin{pmatrix}
-\frac{1}{4} & -\frac{1}{4} & \frac{1}{4\sqrt{2}} & \frac{1}{4\sqrt{2}} & 0 & \frac{1}{4} & \frac{1}{4\sqrt{2}} & \frac{1}{4\sqrt{2}}\\
\frac{1}{8} & -\frac{1}{8} & \frac{1}{4\sqrt{2}} & \frac{1}{4\sqrt{2}} & 0 & 0 & -\frac{1}{4\sqrt{2}} & -\frac{1}{4\sqrt{2}}\\
\frac{1}{4} & \frac{1}{4} & \frac{1}{2\sqrt{2}} & 0 & \frac{1}{2} & 0 & \frac{1}{2\sqrt{2}} & 0\\
0 & 0 & 0 & -\frac{1}{2\sqrt{2}} & 0 & 0 & 0 & \frac{1}{2\sqrt{2}}\\
-\frac{1}{4} & -\frac{1}{4} & \frac{1}{4\sqrt{2}} & -\frac{1}{4\sqrt{2}} & 0 & -\frac{1}{4} & \frac{1}{4\sqrt{2}} & -\frac{1}{4\sqrt{2}}\\
-\frac{1}{8} & \frac{1}{8} & -\frac{1}{4\sqrt{2}} & \frac{1}{4\sqrt{2}} & 0 & 0 & \frac{1}{4\sqrt{2}} & -\frac{1}{4\sqrt{2}}\\
\frac{1}{4} & \frac{1}{4} & 0 & \frac{1}{2\sqrt{2}} & -\frac{1}{2} & 0 & 0 & \frac{1}{2\sqrt{2}}\\
0 & 0 & -\frac{1}{2\sqrt{2}} & 0 & 0 & 0 & \frac{1}{2\sqrt{2}} & 0\\
\frac{1}{4} & \frac{1}{4} & 0 & -\frac{1}{2\sqrt{2}} & -\frac{1}{2} & 0 & 0 & -\frac{1}{2\sqrt{2}}\\
0 & 0 & \frac{1}{2\sqrt{2}} & 0 & 0 & 0 & -\frac{1}{2\sqrt{2}} & 0\\
-\frac{1}{4} & -\frac{1}{4} & -\frac{1}{4\sqrt{2}} & \frac{1}{4\sqrt{2}} & 0 & -\frac{1}{4} & -\frac{1}{4\sqrt{2}} & \frac{1}{4\sqrt{2}}\\
-\frac{1}{8} & \frac{1}{8} & \frac{1}{4\sqrt{2}} & -\frac{1}{4\sqrt{2}} & 0 & 0 & -\frac{1}{4\sqrt{2}} & \frac{1}{4\sqrt{2}}\\
\frac{1}{4} & \frac{1}{4} & -\frac{1}{2\sqrt{2}} & 0 & \frac{1}{2} & 0 & -\frac{1}{2\sqrt{2}} & 0\\
0 & 0 & 0 & \frac{1}{2\sqrt{2}} & 0 & 0 & 0 & -\frac{1}{2\sqrt{2}}\\
-\frac{1}{4} & -\frac{1}{4} & -\frac{1}{4\sqrt{2}} & -\frac{1}{4\sqrt{2}} & 0 & \frac{1}{4} & -\frac{1}{4\sqrt{2}} & -\frac{1}{4\sqrt{2}}\\
\frac{1}{8} & -\frac{1}{8} & -\frac{1}{4\sqrt{2}} & -\frac{1}{4\sqrt{2}} & 0 & 0 & \frac{1}{4\sqrt{2}} & \frac{1}{4\sqrt{2}}\\
\end{pmatrix},\,
\bm{\rho}=
\begin{pmatrix}
a\\
g\\
b\\
c\\
d\\
f\\
h\\
k\\
\end{pmatrix}
\end{equation}
with \(a,g,b,c,d,f,h,k\in\mathbb{R}\).
Since
\begin{equation}
\mathrm{rank}(T)=8,
\end{equation}
we can determine \(\bm{\rho}\) with eight variables from \(\bm{K^{(1)}}\).
\end{proof}
Proposition \ref{状態を区別できる　3D} shows further usefulness of the Kirkwood-Dirac distribution.
We can know from the Kirkwood-Dirac distribution of the \(x\) and \(y\) components of the spin 1 not only the expectation values of the remaining directions of the spin, but also the whole information about the state of the three-state system.
In other words, we need only two observables in order to distinguish the state of the three-state quantum system by the Kirkwood-Dirac distributions while we need eight observables by the expectation values.

The third question is with respect to what state the Kirkwood-Dirac distribution for three-state quantum system becomes real.
\begin{prop}
The expectation value of the z components of the spin of a three-state quantum system becomes zero when the Kirkwood-Dirac distribution of the \(x\) and \(y\) components of the spin 1 for the system is real:
\begin{equation}
^\forall(x,y)\in\mathbb{R}^2,\,K_{\rho}(x,y)\in\mathbb{R}
\implies \langle J_3^{(1)}\rangle_{\rho}=0.
\end{equation}
\end{prop}
\begin{proof}
From Eqs.~(\ref{Kirkwood-Dirac変換}) and (\ref{Kirkwood-Dirac関数　行列表示}), we can easily see
\begin{equation}
^\forall(x,y)\in\mathbb{R}^2,\,K_{\rho}(x,y)\in\mathbb{R}
\iff a=g,\,b=h,\,c=k.
\end{equation}
On the other hand, the expectation value of the \(z\) components of the spin 1 of a three-state quantum system is
\begin{equation}
\langle J_3^{(1)}\rangle_{\rho}=\mathrm{Tr}[J_3^{(1)}\rho]=a-g.
\end{equation}
Thus,
\begin{equation}
\langle J_3^{(1)}\rangle_{\rho}=0 \iff a=g.
\end{equation}
Therefore,
\begin{equation}
^\forall(x,y)\in\mathbb{R}^2,\,K_{\rho}(x,y)\in\mathbb{R}
\implies \langle J_3^{(1)}\rangle_{\rho}=0.
\end{equation}
\end{proof}
It is to be remarked that the inverse of this statement does not necessarily hold.
The expectation value of the \(z\) components of the spin 1 for a state with respect to which the Kirkwood-Dirac distribution is real always becomes zero, but the Kirkwood-Dirac distribution with respect to the state whose expectation value of the \(z\) components of the spin 1 is zero is not always become real.

\subsection{Other observables}
\label{Other Observables in 3D}
In Subsec.~\ref{General Feature of Kirkwood-Dirac Distributions of Spin 1}, we have investigated the features of the Kirkwood-Dirac distribution of the \(x\) and \(y\) components of spin 1.
In this subsection, we will turn our attention to the Kirkwood-Dirac distribution of other observables.
\begin{prop}
\label{他の物理量についてのKirkwood-Dirac関数　3D}
The state of the three-state quantum system can be distinguished by the Kirkwood-Dirac distribution of any non-commutative pair of observables in the spin 1 representation of \(\mathfrak{su}(2)\).
\end{prop}
\begin{proof}
The observables in the spin 1 representation of \(\mathfrak{su}(2)\) is generally given as
\begin{equation}
\hat{A}=\sum_{i=1}^3 a_i\hat{\sigma}_i
\end{equation}
where \(a_i\in\mathbb{R}\) for \(i=1,2,3\) with \(\sum_{i=1}^3a_i=1\).
By the same reasoning as in the proof of Prop.~\ref{他の物理量についてのKirkwood-Dirac関数　2D}, it is sufficient to prove for the pair of observables in the form
\begin{equation}
\hat{A}=\hat{J}_3^{(1)}\cos\varphi+\hat{J}_1^{(1)}\sin\varphi,\;\hat{B}=\hat{J}_3^{(1)},\;\;\;\varphi\neq n\pi,\,n\in\mathbb{Z}
\end{equation}
since the quais-joint-probability distribution is invariant under unitary transformations and the multiplication by scalar only scales the distribution.
The Kirkwood-Dirac distribution generated from the hashed operator
\begin{equation}
\hat{\#}^{\mathrm{K}}_{A,B}=e^{-is(J_3^{(1)}\cos\varphi+J_1^{(1)}\sin\varphi)}e^{-itJ_3^{(1)}}
\end{equation}
is given by
\begin{equation}
K_{\rho}^{\varphi}(x,y)=\sum_{\alpha,\beta=-1,0,+1}K^{\varphi(1)}_{\alpha\beta}\delta(x-\alpha)\delta(x-\beta),
\end{equation}
where the coefficients
\begin{equation}
\bm{K}^{\varphi(1)}=
\begin{pmatrix}
\mathrm{Re}(K^{\varphi(1)}_{++})\\
\mathrm{Im}(K^{\varphi(1)}_{++})\\
\mathrm{Re}(K^{\varphi(1)}_{+0})\\
\mathrm{Im}(K^{\varphi(1)}_{+0})\\
\mathrm{Re}(K^{\varphi(1)}_{+-})\\
\mathrm{Im}(K^{\varphi(1)}_{+-})\\
\mathrm{Re}(K^{\varphi(1)}_{0+})\\
\mathrm{Im}(K^{\varphi(1)}_{0+})\\
\mathrm{Re}(K^{\varphi(1)}_{00})\\
\mathrm{Im}(K^{\varphi(1)}_{00})\\
\mathrm{Re}(K^{\varphi(1)}_{0-})\\
\mathrm{Im}(K^{\varphi(1)}_{0-})\\
\mathrm{Re}(K^{\varphi(1)}_{-+})\\
\mathrm{Im}(K^{\varphi(1)}_{-+})\\
\mathrm{Re}(K^{\varphi(1)}_{-0})\\
\mathrm{Im}(K^{\varphi(1)}_{-0})\\
\mathrm{Re}(K^{\varphi(1)}_{--})\\
\mathrm{Im}(K^{\varphi(1)}_{--})\\
\end{pmatrix}
\end{equation}
are determined by
\begin{equation}
\bm{K^{\varphi(1)}} = T^{\varphi}\bm{\rho}+\bm{K_{0}^{\varphi}}
\end{equation}
by using
\begin{equation}
T^{\varphi}=
\begin{pmatrix}
C_{+} & 0 & Q_{+} & 0 & S/2 & 0 & 0 & 0\\
0 & 0 & 0 & Q_{+} & 0 & S/2 & 0 & 0\\
-S & -S & Q_{+} & 0 & 0 & 0 & Q_{-} & 0\\
0 & 0 & 0 & -Q_{+} & 0 & 0 & 0 & Q_{-}\\
0 & C_{-} & 0 & 0 & S/2 & 0 & Q_{-} & 0\\
0 & 0 & 0 & 0 & 0 &-S/2 & 0 & -Q_{-}\\
S & 0 & -P & 0 & -S & 0 & 0 & 0\\
0 & 0 & 0 & -P & 0 & -S & 0 & 0\\
-C_{0} & -C_{0} &-P & 0 & 0 & 0 & P & 0\\
0 & 0 & 0 & P & 0 & 0 & 0 & P\\
0 & S & 0 & 0 & -S & 0 & P & 0\\
0 & 0 & 0 & 0 & 0 & S & 0 & -P\\
C_{-} & 0 & -Q_{-} & 0 & S/2 & 0 & 0 & 0\\
0 & 0 & 0 & -Q_{-} & 0 & S/2 & 0 & 0\\
-S & -S & -Q_{-} & 0 & 0 & 0 & Q_{+} & 0\\
0 & 0 & 0 & Q_{-} & 0 & 0 & 0 & -Q_{+}\\
0 & C_{+} & 0 & 0 & S/2 & 0 & -Q_{+} & 0\\
0 & 0 & 0 & 0 & 0 & -S/2 & 0 & Q_{+}
\end{pmatrix},\;
\bm{K_{0}^{\varphi}}=
\begin{pmatrix}
0\\
0\\
S\\
0\\
0\\
0\\
0\\
0\\
0\\
C_{0}\\
0\\
0\\
0\\
0\\
S\\
0\\
0\\
0
\end{pmatrix}
\end{equation}
with
\begin{equation}
C_{0}\coloneqq\cos^2\varphi,
\end{equation}
\begin{equation}
C_{+}\coloneqq\Bigl(\frac{1+\cos\varphi}{2}\Bigr)^2,
\end{equation}
\begin{equation}
C_{-}\coloneqq\Bigl(\frac{1-\cos\varphi}{2}\Bigr)^2,
\end{equation}
\begin{equation}
S\coloneqq\frac{\sin^2\varphi}{2},
\end{equation}
\begin{equation}
P\coloneqq\frac{\sin\varphi\cos\varphi}{\sqrt{2}},
\end{equation}
\begin{equation}
Q_{+}\coloneqq\frac{\sin\varphi(1+\cos\varphi)}{2\sqrt{2}},
\end{equation}
\begin{equation}
Q_{-}\coloneqq\frac{\sin\varphi(1-\cos\varphi)}{2\sqrt{2}}.
\end{equation}
Since
\begin{equation}
\mathrm{rank}(T^{\varphi})=8
\end{equation}
whenever \(\varphi\neq n\pi,\,n\in\mathbb{Z}\), we can determine \(\bm{\rho}\) from \(\bm{K^{\varphi(1)}}\).
This is to say that the state of the three-state quantum system \(\rho\) may be recovered from the Kirkwood-Dirac distribution of \(A\) and \(B\), and further from that of any non-commutative pair of observables in the spin 1 representation of \(\mathfrak{su}(2)\). 
\end{proof}
We did not yet find any pairs of observables that do not belong to the spin 1 representation of \(\mathfrak{su}(2)\) but may distinguish the state of the three-state quantum system by their Kirkwood-Dirac distributions.
Proposition \ref{他の物理量についてのKirkwood-Dirac関数　3D} indicates the usefulness of the spin operator when considering its quasi-joint-probability distribution.

\section{Quasi-Probabilities for General Finite-State Quantum System}
\label{Quasi-Probabilities for General Finite-State Quantum System}
In Secs.~\ref{Quasi-Probabilities for Two-State Quantum System} and \ref{Quasi-Probabilities for Three-State Quantum System}, we have seen the features of the quasi-joint-probability distribution for the two-state and three-state quantum systems and found the usefulness of the Kirkwood-Dirac distribution of two directions of spin.
In this section, we will look at the features of the quasi-joint-probability distribution that we can say generally in all finite-state quantum systems.

\subsection{Possible values of observables}
\label{Possible Values of Observables}
First, we will see the problem regarding the possible value of the observables, which we defined in Subsec.~\ref{Finit-State Quantum System and Possible Values of Observables} as the combination of the eigenvalues of the observables in question:
\begin{align*}
(\alpha_1,\alpha_2,...,\alpha_n)\;&\mbox{is a possible value of}\;(A_1,A_2,...,A_n)\\
&\iff \alpha_i\;\mbox{is an eigenvalue of}\;A_i\;\mbox{for}\;i=1,2,...,n.
\stepcounter{equation}\tag{\theequation}
\end{align*}
\begin{prop}
\label{取り得る値　有限次元}
The Kirkwood-Dirac distribution of any combinations of observables with respect to an arbitrary state \(\rho\) of an \(N\)-state quantum system takes non-zero values only at the possible values of the observables:
\begin{equation}
K_{\rho}^{A_1,A_2,...,A_n}(x_1,x_2,...,x_n)\neq0
\end{equation}
only if \((x_1,x_2,...,x_n)\) are the possible values of the observables \((A_1,A_2,...,A_n)\) with \(N,n\in\mathbb{N}\).
\end{prop}
\begin{proof}
All observables for the finite-dimensional quantum system can be diagonalized
\begin{equation}
A_i=U_{i}\,
\mathrm{diag}(\alpha_{i1},\alpha_{i2},...,\alpha_{iN})\,
U^{\dagger}_{i}
\end{equation}
by using unitary matrices \(U_{i}\) for \(i=1,2,...,n\), since observables are Hermitian matrices.
Here, \(\alpha_{ij}\;(j=1,2,...,N)\) are the eigenvalues of \(A_i\) for \(i=1,2,...,n\), where we allow for some \(j_1\neq j_2\), \(\alpha_{ij_1}=\alpha_{ij_2}\).
Therefore, the exponential functions of an observable is given by
\begin{equation}
e^{-is_{i}A_i}=U_{i}\,
\mathrm{diag}(e^{-is\alpha_{i1}},e^{-is\alpha_{i2}},...,e^{-is\alpha_{iN}})\,
U^{\dagger}_{i}.
\end{equation}
By using this, the QJSD of the Kirkwood-Dirac transform is given by
\begin{align*}
\label{Kirkwood-DiracのQJSD　固有値のデルタ関数}
\#^{\mathrm{K}}_{A_1,A_2,...,A_{n}}(x_1,x_2,...,x_n)
&\coloneqq \mathcal{F}^{-1}\Biggl(\prod_{k=1}^n e^{-is_kA_k}\Biggr)(x_1,x_2,...,x_n)\\
&=\mathcal{F}^{-1}\Biggl(\prod_{k=1}^n U_{k}\,
\mathrm{diag}(e^{-is_k\alpha_{k1}},...,e^{-is_k\alpha_{kN}})\,U^{\dagger}_{k}\Biggr)(x_1,...,x_n)\\
&=\prod_{k=1}^n U_{k}\,
\mathrm{diag}\Bigl(\delta(x-\alpha_{k1}),...,\delta(x-\alpha_{kN})\Bigr)\,U^{\dagger}_{k}
\stepcounter{equation}\tag{\theequation}.
\end{align*}
Therefore, every components of the quasi-joint-spectrum distribution (\ref{Kirkwood-DiracのQJSD　固有値のデルタ関数}), thus the Kirkwood-Dirac distribution with respect to an arbitrary state, becomes a linear combination of the delta functions that have the peak at the possible values of the observables:
\begin{equation}
K_{\rho}^{A_1,...,A_n}(x_1,...,x_n)=\prod_{k=1}^{n}\sum_{l=1}^{N}c_{kl}\delta(x_k-\alpha_{kl}),
\end{equation}
with coefficients \(c_{kl}\in\mathbb{C}\;(k=1,...,n,\,l=1,...,N)\) that satisfy the normalization condition \(\prod_{k=1}^{n}\sum_{l=1}^{N}c_{kl}=1\).
\end{proof}
Proposition \ref{取り得る値　有限次元} indicates that the properties of the Kirkwood-Dirac distribution that it becomes a discretized distribution when the possible values of the observable are discretized holds for all finite-state quantum systems.
Because of the way we prove it, we can see that the quasi-joint-probability distributions other than the Kirkwood-Dirac type do not fulfill the properties in Prop.~\ref{取り得る値　有限次元} from the same argument as in Subsec.~\ref{Other Quasi-Classicalizations}.
A similar argument was made by Hofmann \cite{Hofmann2014}, who pointed out that the quasi-probability is limited to the Kirkwood-Dirac distribution under the condition that the quasi-probability \(P_{\psi}(a,b)\) of the observables \(A\) and \(B\) must be zero if the state \(\ket{\psi}\) is the eigenstate of \(A\) of a different eigenvalue \(a'\neq a\) \cite{Hofmann2014}, while our condition is that \(P(a,b)\) must be zero for arbitrary states if \(a\) is not an eigenvalue of \(A\).

\subsection{Restriction to the degeneracy}
\label{Restriction to the Degeneracy}
First, we present a reason why we can distinguish the state of a system by the Kirkwood-Dirac distribution of only two observables.
Since the density operator of the \(N\)-state quantum system has \(N^2-1\) independent components, the quasi-joint-probability distribution must have equal to or more than \(N^2-1\) independent components in order to distinguish the state.
Proposition \ref{取り得る値　有限次元} shows that the Kirkwood-Dirac distribution of \(n\) observables generally has \(2N^n\) parameters.
Although these parameters are not independent of each other because of the restriction such as the condition about marginal probabilities (\ref{probability of one observable}), the Kirkwood-Dirac distribution of \(n\) observables have approximately \(O(N^n)\) independent components.
By comparing the independent components of the density operator and those of the Kirkwood-Dirac distribution, we can roughly see a reason why we only need two observables in order to distinguish the states by the Kirkwood-Dirac distribution of them.

Since the degeneracy of the eigenvalues of the observables in question reduces the parameters of the Kirkwood-Dirac distribution, there exist restrictions to the degeneracy so that the states can be distinguished by the Kirkwood-Dirac distribution of the observables.
The condition for the degeneracy is given as follows.
\begin{prop}
\label{縮退に対する条件}
In order to distinguish the state of an \(N\)-state quantum system by the Kirkwood-Dirac distribution of the observables \(A\) and \(B\), the numbers of their distinct eigenvalues \(N_A\) and \(N_B\) must satisfy the inequality
\begin{equation}
\label{condition to degeneracy　一般}
2N^2-1\leq(2N_A-1)(2N_B-1).
\end{equation}
\end{prop}
\begin{proof}
From Prop.~\ref{取り得る値　有限次元}, the number of locations where the Kirkwood-Dirac distribution of \(A\) and \(B\) takes non-zero values is \(N_AN_B\).
Since the Kirkwood-Dirac distribution has both real and imaginary parts, it has \(2N_AN_B\) components.
The restriction that the marginal probability distribution should be real gives \(N_A+N_B-1\) conditions and that the sum of the quasi-probability distribution over the whole range should be unity gives one condition.
Since the number of independent components of the Kirkwood-Dirac distribution should be equal to or more than the number of that of the density operator in order for the state to be distinguished by the Kirkwood-Dirac distribution, the inequality
\begin{equation}
N^2-1\leq 2N_AN_B-N_A-N_B,
\end{equation}
or equivalently,
\begin{equation}
2N^2-1\leq(2N_A-1)(2N_B-1).
\end{equation}
must be satisfied.
\end{proof}
As a special case of Prop.~\ref{縮退に対する条件}, the following corollaries hold.
\begin{cor}
\label{縮退に関する条件　両方同じ}
In order to distinguish the state of an \(N\)-state quantum system by the Kirkwood-Dirac distribution of two observables which have the same number \(N'\) of distinct eigenvalues, the inequality
\begin{equation}
\label{condition to degeneracy 両方同じ}
N'\geq\frac{1}{2}\Bigl(\sqrt{2N^2-1}+1\Bigr)
\end{equation}
must be satisfied.
\end{cor}
\begin{cor}
\label{縮退に対する条件　片方}
In order to distinguish the state of an \(N\)-state quantum system by the Kirkwood-Dirac distribution of a non-degenerated observable \(A\) and an observable \(B\) which have \(N'\) distinct eigenvalues, the inequality
\begin{equation}
\label{condition to degeneracy 片方}
N'\geq \frac{N^2+N-1}{2N-1}
\end{equation}
must be satisfied.
\end{cor}
Since Ineq.~(\ref{condition to degeneracy 片方}) becomes \(N'\geq5/3\) in the case of \(N=2\) and \(N'\geq11/5\) in the case of \(N=3\), no degeneracy is allowed for both the cases because \(N'\leq N\) and \(N'\in\mathbb{N}\).
It is to be remarked that they are only the necessity conditions.
For example, although the reducible representation of the spins
\begin{equation}
A=
\begin{pmatrix}
0 & 1 & 0\\
1 & 0 & 0\\
0 & 0 & 0
\end{pmatrix},\;\;\;
B=
\begin{pmatrix}
0 & -i & 0\\
i & 0 & 0\\
0 & 0 & 0
\end{pmatrix}
\end{equation}
satisfy the degeneracy condition (\ref{condition to degeneracy　一般}), the state of the three-state quantum system cannot be distinguished by their Kirkwood-Dirac distribution.

\section{Summary and Discussion}
\label{Summary and Discussion}
Through the above discussions, we have investigated the features of the quasi-joint-probability distributions for the finite-state quantum systems and showed the advantages of the Kirkwood-Dirac distribution.
We discovered that the Kirkwood-Dirac distribution takes non-zero values only at the possible values in the same manner as the genuine probability distribution, which is in contrast to most of the other quasi-joint-probability distributions.
We also proved for the two-state and three-state quantum systems that the states of the system can be distinguished by the Kirkwood-Dirac distribution of only two directions of the spin.
In addition, we discovered that the non-Hermicity of the quasi-joint-spectrum distribution is the necessary and sufficient condition for the states of the system to be distinguished by a quasi-joint-probability distribution of the \(x\) and \(y\) components of the spin in the two-state quantum system.
These facts support the usefulness of the Kirkwood-Dirac distribution for analyzing the state of a quantum system.

Still, there remains the question of whether the similar property of the statements about the Kirkwood-Dirac distribution for the two-state and the three-state quantum systems generally hold in systems described by the higher-dimensional Hilbert spaces.
It is also worth investigating how these discussions are related to the quasi-probability for infinite-state systems, such as the familiar optical systems.

\section*{Acknowledgement}
S.U. acknowledges the opportunity provided by the Undergraduate Research Opportunity Program (UROP) of Institute of Industrial Science, The University of Tokyo.
The present work is partially supported by JSPS KAKENHI Grant Numbers JP19H00658, JP21H01005, JP22H01140 and JP22K13970.

\bibliographystyle{unsrt}
\bibliography{umekawa}

\end{document}